\documentclass{lmcs} %%% last changed 2014-08-20
\pdfoutput=1
\usepackage[utf8]{inputenc}

\usepackage{lastpage}
\lmcsdoi{19}{4}{15}
\lmcsheading{}{\pageref{LastPage}}{}{}%
{Oct.~29,~2021}{Nov.~27,~2023}{}

\usepackage{amsmath}
\usepackage{amssymb}
\usepackage{amsfonts}
\usepackage{amsthm}
\usepackage{multirow}
\usepackage{booktabs}

\usepackage{wasysym}
\usepackage{array}

\newcommand{\DBfont}[1]{\mathcal{#1}}

\newcommand{\dataCom}[1]{\bigO^{\textsf{dc}}(#1)}

\newcommand{\DBtimesRE}[2]{G_{\boxtimes}(#1, #2)}
\newcommand{\DBtimesREV}[2]{V_{\boxtimes}(#1, #2)}
\newcommand{\DBtimesREE}[2]{E_{\boxtimes}(#1, #2)}

\DeclareMathOperator{\boole}{\textsf{B}}

\DeclareMathOperator{\degree}{\Delta}
\DeclareMathOperator{\avgdegree}{\overline{\degree}}

\DeclareMathOperator{\RE}{\mathsf{RE}}

\DeclareMathOperator{\poly}{\mathsf{poly}}

\DeclareMathOperator{\OMv}{\textsf{OMv}}
\DeclareMathOperator{\OV}{\textsf{OV}}

\DeclareMathOperator{\TriProb}{\textsf{Triangle}}
\DeclareMathOperator{\combBMMProb}{\textsf{com-BMM}}

\DeclareMathOperator{\DBD}{\DBfont{D}}

\DeclareMathOperator{\npclass}{\mathsf{NP}}

\DeclareMathOperator{\pspaceclass}{\mathsf{PSPACE}}

\DeclareMathOperator{\bigO}{O}

\DeclareMathOperator{\smallO}{o}

\DeclareMathOperator{\evalProb}{\textsc{RPQ-Eval}}
\DeclareMathOperator{\witnessProb}{\textsc{RPQ-Witness}}
\DeclareMathOperator{\booleProb}{\textsc{RPQ-Boole}}
\DeclareMathOperator{\enumProb}{\textsc{RPQ-Enum}}
\DeclareMathOperator{\checkProb}{\textsc{RPQ-Test}}
\DeclareMathOperator{\countProb}{\textsc{RPQ-Count}}

\newcommand{\enumProbClass}[1]{\enumProbShort(#1)}

\DeclareMathOperator{\evalProbShort}{\textsc{Eval}}
\DeclareMathOperator{\witnessProbShort}{\textsc{Witness}}
\DeclareMathOperator{\booleProbShort}{\textsc{Boole}}
\DeclareMathOperator{\enumProbShort}{\textsc{Enum}}
\DeclareMathOperator{\checkProbShort}{\textsc{Test}}
\DeclareMathOperator{\countProbShort}{\textsc{Count}}

\DeclareMathOperator{\approxEnumProb}{\textsc{App-RPQ-Enum}}

\newcommand{\RPQ}{\mathsf{RPQ}}
\newcommand{\CRPQ}{\mathsf{CRPQ}}

\newcommand{\CQ}{\mathsf{CQ}}

\newcommand{\BTRPQ}{\mathsf{BT\text{-}RPQ}}
\newcommand{\SRPQ}{\mathsf{S\text{-}RPQ}}

\DeclareMathOperator{\BMMProb}{\textsf{BMM}}

\DeclareMathOperator{\SBMMProb}{\textsf{SBMM}}

\DeclareMathOperator{\init}{0}

\DeclareMathOperator{\NFA}{\mathsf{NFA}}

\DeclareMathOperator{\eword}{\varepsilon}

\DeclareMathOperator{\lang}{\mathcal{L}}

\DeclareMathOperator{\altop}{\vee}

\DeclareMathOperator{\preprocess}{\mathsf{preprocess}}

\DeclareMathOperator{\update}{\mathsf{update}}
\DeclareMathOperator{\enum}{\mathsf{enum}}

\newcommand{\ta}{\ensuremath{\mathtt{a}}}
\newcommand{\tb}{\ensuremath{\mathtt{b}}}
\newcommand{\tc}{\ensuremath{\mathtt{c}}}
\newcommand{\td}{\ensuremath{\mathtt{d}}}

\keywords{Graph Databases, Regular Path Queries, Enumeration, Fine-Grained Complexity}

\usepackage{hyperref}

\begin{document}

\title[Fine-Grained Complexity of Regular Path Queries]{Fine-Grained Complexity of Regular Path Queries}
\titlecomment{This is the full version of the article~\cite{CaselSchmid2021}.
The first author has been funded by the Federal Ministry of Education and Research of Germany (BMBF) in the KI-LAB-ITSE framework -- project number 01IS19066. The second author has been funded by the German Research Foundation (Deutsche Forschungsgemeinschaft, DFG) -- project number 416776735 (gef\"ordert durch die Deutsche Forschungsgemeinschaft (DFG) -- Projektnummer 416776735).\\
\indent \emph{2012 ACM Subject Classification}: Theory of computation $\rightarrow$ Regular languages; Theory of computation $\rightarrow$ Problems, reductions and completeness; Theory of computation $\rightarrow$ Database query languages (principles); Theory of computation $\rightarrow$ Data structures and algorithms for data management}

\author[K. Casel]{Katrin Casel\lmcsorcid{0000-0001-6146-8684}}[a]	

\author[M. L.\ Schmid]{Markus L.\ Schmid\lmcsorcid{0000-0001-5137-1504}}[b]

\address{Hasso Plattner Institute, University of Potsdam, Potsdam, Germany}	
\email{Katrin.Casel@hpi.de}

\address{Humboldt-Universit\"at zu Berlin, Unter den Linden 6, D-10099, Berlin, Germany}	
\email{MLSchmid@MLSchmid.de}

\begin{abstract}
\noindent A regular path query (RPQ) is a regular expression $q$ that returns all node pairs $(u, v)$ from a graph database that are connected by an arbitrary path labelled with a word from $L(q)$. The obvious algorithmic approach to RPQ-evaluation (called PG-approach), i.\,e., constructing the product graph between an NFA for $q$ and the graph database, is appealing due to its simplicity and also leads to efficient algorithms. However, it is unclear whether the PG-approach is optimal. We address this question by thoroughly investigating which upper complexity bounds can be achieved by the PG-approach, and we complement these with conditional lower bounds (in the sense of the fine-grained complexity framework). A special focus is put on enumeration and delay bounds, as well as the data complexity perspective. A main insight is that we can achieve optimal (or near optimal) algorithms with the PG-approach, but the delay for enumeration is rather high (linear in the database). We explore three successful approaches towards enumeration with sub-linear delay: super-linear preprocessing, approximations of the solution sets, and restricted classes of RPQs.
\end{abstract}

\maketitle

\section{Introduction}

An essential component of graph query languages (to be found both in academical prototypes as well as in industrial solutions) are \emph{regular path queries} ($\RPQ$s). Abstractly speaking, a regular expression $q$ over some alphabet $\Sigma$ is interpreted as query that returns from a $\Sigma$-edge-labelled, directed graph $\DBD$ (i.\,e., a \emph{graph database}) the set $q(\DBD)$ of all node pairs $(u, v)$ that are connected by a $q$-path, i.\,e., a path labelled with a word from $q$'s language (and possibly also a witness path per node pair, or even all such paths). This simple, yet relevant concept has heavily been studied in database theory (the following list is somewhat biased towards recent work): results on $\RPQ$s~\cite{CruzEtAl1987,BaganEtAl2020,LosemannMartens2013,MartensEtAl2020,MartensTrautner19_2,BienvenuThomazo2016}, conjunctive $\RPQ$s~\cite{BienvenuEtAl2015,ReutterEtAl17,BagetEtAl2017} and extensions thereof~\cite{LibkinEtAl2016,BarceloEtAl2012,LibkinEtAl2016,FreydenbergerSchweikardt2013}, questions of static analysis~\cite{FigueiraEtAl2020,BarceloEtAl2019,Figueira2020,GluchEtAl2019,RomeroEtAl2017}, experimental analyses~\cite{BonifatiEtAl2017,BonifatiEtAl2020,MartensTrautner2019}, and surveys of this research area~\cite{Barcelo2013,Wood2012,AnglesEtAl2017,CalvaneseEtAl2003}.\par
In the simplest setting, where we are only interested in the node pairs (but no paths) connected by  \emph{arbitrary} $q$-paths (instead of, e.\,g., simple paths), evaluation can be done efficiently. Deviating from this simple setting, however, leads to intractability: if we ask for nodes connected by \emph{simple paths} (no repeated nodes), or connected by \emph{trails} (no repeated arcs), then $\RPQ$ evaluation is $\npclass$-hard even in data-complexity (see~\cite{MendelzonWood1995, BaganEtAl2020, MartensTrautner2018}~and~\cite{MartensEtAl2020}, respectively). Note that the simple path and trail semantics are mostly motivated by the fact that under these semantics there is only a finite number of $q$-paths per node pair. If we move to conjunctions of $\RPQ$s ($\CRPQ$s) or even more powerful extensions motivated by practical requirements, then also with the arbitrary path semantics evaluation becomes intractable in combined complexity (i.\,e., they inherit hardness from relational conjunctive queries ($\CQ$s)). \par
In order to guide practical developments in the area of graph databases, the computational hard cases of $\RPQ$ (and $\CRPQ$) evaluation have been thoroughly investigated in database theory. However, with respect to arbitrary $q$-paths, research seems to have stopped at the conclusion that efficient evaluation is possible by the following simple \emph{PG-approach}: given graph database $\DBD = (V_{\DBD}, E_{\DBD})$ and $\NFA$ $M_q$ for $q$ with state set $V_q$, construct the \emph{product graph} $G_{\DBD, q}$ with nodes $V_{\DBD} \times V_q$ and an arc $((u, p), (v, p'))$ iff, for some $\ta \in \Sigma$, $\DBD$ has an arc $(u, \ta, v)$ and $M_q$ has an arc $(p, \ta, p')$, and then use simple graph-searching techniques on $G_{\DBD, q}$. \par
The PG-approach is explicitly defined in several papers, e.\,g.,~\cite{MendelzonWood1995, Barcelo2013, MartensTrautner2018}, and mainly used to prove a worst-case upper bound (actually $\bigO(|q||\DBD|)$ for Boolean evaluation; in~\cite{MartensTrautner2018} it is used for enumerating $q$-paths between two given nodes. But it is also very appealing from a practical point of view due to its simplicity: we are just coupling well-understood algorithmic concepts like finite automata and graph reachability algorithms. Arguably, implementing the PG-approach is an exercise suitable for a first year programming course (making it feasible and cost-efficient for industrial systems). As it seems, putting $\RPQ$ evaluation with arbitrary path semantics and the respective PG-approach into the focus of a thorough theoretical study has not yet been done. This paper is devoted to this task. In particular, we wish to investigate the following two (somewhat overlapping) aspects:

\begin{enumerate}
\item Applicability of the PG-approach: the PG-approach is suited for solving simple evaluation problems like checking $q(\DBD) = \emptyset$ or $(u, v) \in q(\DBD)$ (for given $u, v \in V_{\DBD}$), but is it also appropriate for more relevant tasks like computing, counting or enumerating $q(\DBD)$?
\item Optimality of the PG-approach: Does the PG-approach lead to optimal algorithms, or can it be beaten by conceptionally different techniques?
\end{enumerate}

Answering these questions provides a better theoretical understanding of $\RPQ$-evaluation (which, as mentioned above, are at the heart of many graph query languages). But also for the more powerful $\CRPQ$s and more practically motivated graph query languages, we can derive valuable insights from our investigation. Let us mention two such examples (a complete summary of our results follows further down). As noted in~\cite{Barcelo2013}, we can reduce $\CRPQ$ evaluation to the evaluation of relational $\CQ$s by first constructing all tables represented by the single $\RPQ$s and then evaluating a $\CQ$ over this database. To do this, we first have to compute the results of all $\RPQ$s, so it seems helpful to know the best algorithms for this intermediate task. Moreover, if we want to benefit from the existing $\CQ$ evaluation techniques (e.\,g., exploiting acyclicity etc.) we are more or less forced to this two-step approach. With respect to \emph{enumerations} of $\CQ$s, it is known that linear preprocessing and constant delay enumeration is possible provided that the $\CQ$s satisfy certain acyclicity properties (see~\cite{BaganEtAl2007, BerkholzEtAl2017}, or the surveys~\cite{BerkholzEtAl2020, Segoufin2015}). Unfortunately, these techniques do not carry over to $\CRPQ$s since, as we show, linear preprocessing and constant delay enumeration is not possible even for single $\RPQ$s (conditional to some complexity assumptions).

Since the problem we investigate can be solved in polynomial time (also in combined complexity), we cannot show lower bounds in terms of hardness results for complexity classes like $\npclass$ or $\pspaceclass$. Instead, we make use of the framework of \emph{fine-grained complexity}, which allows to prove lower bounds that are conditional on some algorithmic assumptions (see the surveys~\cite{Williams2015,Bringmann2019,Williams2018_2}). In particular, fine-grained complexity is a rather successful toolbox for giving evidence that the \emph{obvious} algorithmic approach to some basic problem, is also the \emph{optimal} one. This is exactly our setting here, with respect to $\RPQ$-evaluation and the PG-approach. To the knowledge of the authors, such conditional lower bounds are not yet a well-established technique in database theory (however, see~\cite[Section~6]{BerkholzEtAl2020} for a survey of conditional lower bounds in the context of $\CQ$ enumeration). \par
A main challenge is that fine-grained complexity is not exactly tailored to either the data-complexity perspective or to enumeration problems. We will next outline our results.

\subsection{Our Contribution}

\begin{table}

\renewcommand{\arraystretch}{1.2}    % Relaxed table
\begin{tabular}{|l|l|l|l|l|}\hline
\multicolumn{2}{|l|}{\textbf{Non-enum. Results}}  &  \begin{minipage}[c][0.8cm][c]{3cm}$\booleProbShort$, $\checkProbShort$, $\witnessProbShort$\end{minipage} & $\evalProbShort$ & $\countProbShort$ \\\noalign{\hrule height 1.5pt}

\multicolumn{2}{|l|}{\begin{minipage}[t][0.8cm][t]{2cm}upper bounds\end{minipage}} & $\bigO(|\DBD||q|)$ & \begin{minipage}[t][0.8cm][t]{3cm}$\bigO(|V_{\DBD}||\DBD||q|)$\\$\bigO((|V_{\DBD}||q|)^{\omega})$\end{minipage} & \begin{minipage}[t][0.8cm][t]{2.7cm}$\bigO(|V_{\DBD}||\DBD||q|)$\\$\bigO((|V_{\DBD}||q|)^{\omega})$\end{minipage}\\\noalign{\hrule height 1.5pt}

\multirow{3}{1.4cm}{\begin{minipage}{1cm}lower\\ bounds\end{minipage}} & \begin{minipage}{2.6cm} $\OV \& \combBMMProb$ \end{minipage} & \begin{minipage}{3cm}$\bigO((|\DBD||q|)^{1-\epsilon})$ \end{minipage} & --- & ---\\\cline{2-5}

& $\OV$ & --- &  ---  & \begin{minipage}{3cm}$\dataCom{(|V_{\DBD}||\DBD|)^{1-\epsilon}}$\end{minipage} \\\cline{2-5}

& $\SBMMProb$ & --- & \begin{minipage}{3cm}$\dataCom{|q(\DBD)| + |\DBD|}$ \end{minipage} & --- \\\cline{2-5}

& \begin{minipage}{1.9cm}$\combBMMProb$\end{minipage} &  --- & \begin{minipage}{2.7cm}$\dataCom{(|V_{\DBD}||\DBD|)^{1-\epsilon}}$ \end{minipage} & --- \\\hline

\end{tabular}
\caption{All upper bounds can be achieved as running times of some algorithm, while the lower bounds cannot be achieved as running time by any algorithm, unless the displayed hypothesis fails. The exponent $\omega$ denotes the best known matrix multiplication exponent.}
\label{ResultsTableNonEnum}
\end{table}

All investigated $\RPQ$-evaluation problems are summarised on page~\pageref{ProblemsTable} (see especially Table~\ref{ProblemsTable}). In the following, $\DBD = (V_{\DBD}, E_{\DBD})$ is the graph database, $q$ is the $\RPQ$, and $\epsilon > 0$. With the notation $\dataCom{\cdot}$, we hide factors $f(|q|)$ for some function $f$ (i.\,e., it is used for stating data-complexities). All lower bounds mentioned in the following are conditional to some of the algorithmic assumptions summarised in Section~\ref{sec:fineGrainedComplexity} (we encourage the reader less familiar with fine-grained complexity hypotheses to have a look at this section first, which can be read independently). For presentational reasons, we do not always explicitly mention this in the rest of the introduction and when we say that a certain running time is ``not possible'', this statement is always conditional in this sense (see Tables~\ref{ResultsTableNonEnum}~and~\ref{ResultsTableEnum} for the actual hypotheses). As common in fine-grained complexity, we rule out \emph{true} sub-linear ($\bigO(n^{1-\epsilon})$), sub-quadratic ($\bigO(n^{2-\epsilon})$), or sub-cubic ($\bigO(n^{3-\epsilon})$) running times, but not possible running time improvements by logarithmic factors, e.\,g., $\bigO(\frac{n^{3}}{\log(n)})$.

\subsubsection{Non-Enumeration Variants} 

The following results are summarised in Table~\ref{ResultsTableNonEnum}. For the simple problems $\booleProbShort$ (checking $q(\DBD) = \emptyset$), $\checkProbShort$ (checking $(u, v) \in q(\DBD)$) and $\witnessProbShort$ (computing some element from $q(\DBD)$), the PG-approach yields an upper bound of $\bigO(|\DBD||q|)$, which is optimal (since linear) in data complexity, and we can show lower bounds demonstrating its optimality also in combined complexity. For $\evalProbShort$ (computing the set $q(\DBD)$) the PG-approach yields a data complexity upper bound of $\dataCom{|V_{\DBD}||\DBD|}$, which cannot be improved by \emph{combinatorial} algorithms, although $\dataCom{|V_{\DBD}|^{2.37}}$ is possible by fast matrix multiplication (see Section~\ref{sec:fineGrainedComplexity} for a discussion of the meaning of the term ``combinatorial''). In addition, we can show that linear time data complexity, i.\,e., $\dataCom{|q(\DBD)| + |\DBD|}$, is not possible even for non-combinatorial algorithms. For $\countProbShort$ (computing $|q(\DBD)|$), we get $\dataCom{|V_{\DBD}||\DBD|}$ as upper and lower bound, not restricted to combinatorial algorithms.

\subsubsection{Enumeration}

Our results for $\RPQ$-enumeration are summarised in Table~\ref{ResultsTableEnum}. An entry ``$\bigO(\text{delay})$'' in column ``preprocessing'' means that the preprocessing is bounded by the delay (which means that no preprocessing is required). The column ``sorted'' indicates whether the enumeration is produced lexicographically sorted.\par
In comparison to the non-enumeration problem variants, the picture is less clear and deserves more explanation. The PG-approach yields a simple enumeration algorithm with delay $\bigO(|\DBD||q|)$, that also trivially supports updates in constant time, since the preprocessing fits into the delay bound. Our lower bounds for $\booleProbShort$ also mean that this delay cannot be improved in terms of \emph{combined complexity}. While this lower bound was interesting for problems like $\booleProbShort$ etc., it now gives a correct answer to the wrong question. The main goal now should be to find out whether we can remedy the linear dependency of the delay on $|\DBD|$, at the expense of spending more time in terms of $|q|$, or of losing the ability of handling updates, or even of allowing a slightly super-linear preprocessing.\par
In this regard, the strongest result would be linear preprocessing $\bigO(|\DBD|f(|q|))$ and constant delay $\bigO(f(|q|))$. However, we can rule this out even for algorithms not capable of handling updates. Then, the next question is which non-constant delays can be achieved that are strictly better than linear. For example, none of our lower bounds for the non-enumeration variants suggest that linear preprocessing and a delay bounded by, e.\,g., $|V_{\DBD}|$ or the degree of $\DBD$, should not be possible. We are \emph{not} able to answer this question in its general form (and believe it to be very challenging), but we are able to provide several noteworthy insights. \par
For linear preprocessing, a delay of $\bigO(|V_{\DBD}|)$ (if possible at all) cannot be beaten by combinatorial algorithms (even without updates). This can be strengthened considerably, if we also require updates in some reasonable time: for general algorithms (i.\,e., not necessarily combinatorial) delay \emph{and} update time strictly better than $\dataCom{|V_{\DBD}|}$ is not possible even with arbitrary preprocessing, and for combinatorial algorithms with linear preprocessing even delay \emph{and} update time of $\bigO(|\DBD|)$ cannot be beaten.
This last result nicely complements the upper bound at least for combinatorial algorithms and in the dynamic case.\par
In summary, for linear preprocessing, $\bigO(|V_{\DBD}|)$ is a lower bound for the delay and if we can beat $\bigO(|\DBD|)$, we should not be able to also support updates.

\begin{table}
\renewcommand{\arraystretch}{1.2}
\begin{tabular}{|l|l|l|l|l|l|}\hline

\multicolumn{2}{|l|}{\textbf{Enum. Results}} & \multicolumn{4}{l|}{$\enumProbShort$} \\\cline{3-6}

\multicolumn{2}{|l|}{\begin{minipage}{2cm}\end{minipage}} & \begin{minipage}{2.7cm}preprocessing\end{minipage} & \begin{minipage}{2.4cm}delay\end{minipage} & sorted & updates\\\noalign{\hrule height 1.5pt} 

\multicolumn{2}{|l|}{\begin{minipage}[t][0.8cm][t]{2cm}upper bounds\end{minipage}} & \begin{minipage}{2cm}$\bigO(\text{delay})$\end{minipage} & \begin{minipage}{2cm}$\bigO(|\DBD||q|)$\end{minipage} & \begin{minipage}{1cm}\checked \end{minipage}&  \begin{minipage}{2.35cm}$\bigO(1)$\end{minipage} \\\noalign{\hrule height 1.5pt}

\multirow{3}{1.3cm}{\begin{minipage}{1cm}lower\\ bounds\end{minipage}} & \begin{minipage}{2.8cm} $\OV \& \combBMMProb$ \end{minipage} &  \begin{minipage}{2cm}$\bigO(\text{delay})$\end{minipage}&\begin{minipage}{2cm}$\bigO((|\DBD||q|)^{1-\epsilon})$ \end{minipage} & $\times$ & $\times$\\\cline{2-6}

& $\SBMMProb$ & \begin{minipage}{1cm}$\dataCom{|\DBD|}$\end{minipage} &  $\dataCom{1}$ & $\times$ & $\times$\\\cline{2-6}

& \begin{minipage}{1cm}$\combBMMProb$\end{minipage}  & $\dataCom{|\DBD|}$ &  $\dataCom{|V_{\DBD}|^{1-\epsilon}}$ & $\times$ & $\times$ \\\cline{2-6}

& $\OMv$ & arbitrary & $\dataCom{|V_{\DBD}|^{1-\epsilon}}$ & $\times$ & $\dataCom{|V_{\DBD}|^{1-\epsilon}}$\\\cline{2-6}

& $\combBMMProb$ & $\dataCom{|\DBD|}$ & $\dataCom{|V_{\DBD}|^{2-\epsilon}}$ & $\times$ & $\dataCom{|V_{\DBD}|^{2-\epsilon}}$\\\hline

\end{tabular}
\caption{All upper bounds can be achieved as running times of some algorithm, while the lower bounds cannot be achieved as running time by any algorithm, unless the displayed hypothesis fails. The exponent $\omega$ denotes the best known matrix multiplication exponent.}
\label{ResultsTableEnum}
\end{table}

\subsubsection{Enumeration of Restricted Variants} 

Finally, we obtain restricted problem variants that can be solved with delay strictly better than $\bigO(|\DBD|)$ (in data complexity). We explore three different approaches: 
\begin{enumerate}
\item by allowing super-linear preprocessing of $\dataCom{\avgdegree(\DBD) \log(\avgdegree(\DBD))|\DBD|}$ (where $\avgdegree(\DBD)$ is the average degree of $\DBD$), we can achieve a delay of $\bigO(|V_{\DBD}|)$; 
\item in linear preprocessing and constant delay, we can enumerate a representative subset of $q(\DBD)$ instead of the whole set $q(\DBD)$; 
\item for a subclass of $\RPQ$s, we can solve $\enumProb$ with linear preprocessing and delay $\bigO(\degree(\DBD))$ (where $\degree(\DBD)$ is the maximum degree of $\DBD$). 
\end{enumerate}

\section{Main Definitions}\label{sec:mainDefs}

Let $\mathbb{N} = \{1, 2, 3, \ldots\}$ and $[n] = \{1, 2, \ldots, n\}$ for $n \in \mathbb{N}$. For a finite alphabet $A$, $A^+$ denotes the set of non-empty words over $A$ and $A^* = A^+ \cup \{\eword\}$ (where $\eword$ is the empty word). For a word $w \in A^*$, $|w|$ denotes its length; $w^1 = w$ and $w^k = w w^{k-1}$ for every $k \geq 2$. For $L, K \subseteq A^*$, let $L \cdot K = \{w_1 \cdot w_2 \mid w_1 \in L, w_2 \in K\}$, let $L^1 = L$ and $L^k = L \cdot L^{k-1}$ for every $k \geq 2$, let $L^+ = \bigcup_{k \geq 1} L^k$ and $L^* = L^+ \cup \{\eword\}$.

\subsection{\texorpdfstring{$\Sigma$}{Σ}-Graphs} 

We now define the central graph model that is used to represent graph databases as well as finite automata. Let $\Sigma$ be a finite alphabet of constant size. A \emph{$\Sigma$-graph} is a directed, edge labelled multigraph $G = (V, E)$, where $V$ is the set of \emph{vertices} (or \emph{nodes}) and $E \subseteq V \times (\Sigma \cup \{\eword\}) \times V$ is the set of \emph{edges} (or \emph{arcs}). For $u  \in V$ and $x \in \Sigma \cup \{\eword\}$, $E_x(u) = \{v \mid (u, x, v) \in E\}$ is the set of \emph{$x$-successors of $u$}. A path from $w_0 \in V$ to $w_k \in V$ of length $k \geq 0$ is a sequence $p = (w_0, a_1, w_1, a_2, w_2 \ldots, w_{k-1}, a_{k}, w_k)$ with $(w_{i-1}, a_{i}, w_{i}) \in E$ for every $i \in [k]$. We say that $p$ is \emph{labelled} with the \emph{word} $a_1 a_2 \ldots a_k \in \Sigma^*$. According to this definition, for every $v \in V$, $(v)$ is a path from $v$ to $v$ of length $0$ that is labelled by $\eword$. Hence, every node $v$ of a $\Sigma$-graph has an $\eword$-labelled path to itself, even though there might not be an $\eword$-arc from $v$ to $v$. Moreover, due to $\eword$ as a possible edge-label, paths of length $k$ may be labelled with words $w$ with $|w| < k$. The size of $G = (V, E)$ is $|G| = \max\{|V|, |E|\}$.  \par
For any $\Sigma$-graph $G = (V, E)$, we call $(V, \{(u, v) \mid u \neq v \wedge \exists x \in \Sigma \cup \{\eword\}: (u, x, v) \in E\})$ the \emph{underlying graph} of $G$ (note that the underlying graph is simple, non-labelled and has no loops). In particular, by a slight abuse of notation, we denote by $E^*$ the reflexive-transitive closure of the underlying graph of $G$. Since we always assume $|\Sigma|$ to be a constant, we have that $|G| = \Theta(\max\{|V|, |\{(u, v) \mid u \neq v \wedge \exists x \in \Sigma \cup \{\eword\}: (u, x, v) \in E\}|\})$ (i.\,e., $|G|$ is asymptotically equal to the size of its underlying graph). For every $u \in V$, the \emph{degree of $u$} is $\degree(u) = |\bigcup_{x \in \Sigma \cup \{\eword\}}E_x(u)|$ (so $\degree(u)$ is actually the out-degree), and the \emph{maximum degree} of $G$ is $\degree(G) = \max\{|\degree(u)| \mid u \in V\}$. The \emph{average degree} of $G$ is $\avgdegree(G) = \tfrac{1}{|V|}\sum_{u \in V} |\degree(u)|$. Obviously, $\avgdegree(G) \leq \degree(G) \leq |V|$.

Since $\Sigma$-graphs are the central data structures for our algorithms, we have to discuss implementational aspects of $\Sigma$-graphs in more detail. The set $V$ of a $\Sigma$-graph $G = (V, E)$ is represented as a list, and, for every $u \in V$ and for every $x \in \Sigma \cup \{\eword\}$, we store a list of all $x$-successors of $u$, which is called the \emph{$x$-adjacency list} for $u$. We assume that we can check in constant time whether a list is empty and we can insert elements in constant time. However, finding and deleting an element from a list requires linear time. Furthermore, we assume that we always store together with a node a pointer to its adjacency list (thus, we can always retrieve the $x$-adjacency list for a given node in constant time).

\begin{rem}
All lower bounds presented in this paper hold for any graph representation that can be constructed in time linear in $|G| = \max\{|V|, |E|\}$. For the upper bounds, we chose the simple representation with adjacency lists as it emerged as the natural structure for our enumeration approach; let us point out here that since we always store pointers to the adjacency lists along with the nodes, we can perform a breadth-first search (BFS) from any given start node $u$ in time $\bigO(|G|)$.  It is a plausible assumption that most specific graph representations can be transformed into our list-based representation without much effort. This ensures a certain generality of our upper complexity bounds in the sense that the corresponding algorithms are, to a large extent, independent from implementational details. Note also that the list-based structure only requires space linear in $|G|$.\par
In the adjacency list representation, we do not have random access to specific nodes in the graph database, or to specific neighbours of a given node. Thus, we have to measure a non-constant running-time for performing such operations. However, the algorithms for our upper bounds are independent from this aspect, i.\,e., the total running times would not change if we assume random access to nodes in constant time. \par
An exception to this is Theorem~\ref{crossProductEnumSuperlinearPreprocTheorem}, for which we can obtain some small improvement by applying the technique of lazy array initialization (see Remark~\ref{lazyStuff}).
\end{rem}

For a $\Sigma$-graph $G = (V, E)$, we denote by $G^R = (V, E^R)$ the $\Sigma$-graph obtained from $G$ by reversing all arcs, i.\,e., $E^R = \{(v, x, u) \mid (u, x, v) \in E\}$. 

\begin{lem}\label{reversalLemma}
Let $G = (V, E)$ be a $\Sigma$-graph. Then $G^R$ can be computed in time $\bigO(|G|)$.
\end{lem}

\begin{proof}
Since $G^R = (V, E^R)$, it is sufficient to show how the adjacency lists can be computed that represent $E^R$. To avoid confusion with respect to whether we talk about the $\Sigma$-graph $G$ or the $\Sigma$-graph $G^R$ to be constructed, we denote the $x$-adjacency lists by $x$-$G$-adjacency lists or $x$-$G^R$-adjacency lists, respectively. \par 
We first move through the list for $V$ and, for every $u \in V$, we store this node in an array along with an empty $x$-$G^R$-adjacency list for every $x \in \Sigma \cup \{\eword\}$. This requires time $\bigO(|V|)$. Then, for every $u \in V$ and $x \in \Sigma \cup \{\eword\}$, we move through the $x$-$G$-adjacency list for $u$, and for every element $v$ that we encounter, we add $u$ to the $x$-$G^R$-adjacency list for $v$. 
Since we can access all $x$-$G^R$-adjacency lists in constant time, we only have to add as many elements to some $x$-$G^R$-adjacency as there are edges in $E$, i.\,e., the second step can be done in time $\bigO(|E|)$. Consequently, we can construct all $x$-$G^R$-adjacency lists in total time $\bigO(|V| + |E|) = \bigO(|G|)$.
\end{proof}

\subsection{Graph Databases and Regular Path Queries} 

A \emph{nondeterministic finite automaton} ($\NFA$ for short) is a tuple $M = (G, S, T)$, where $G = (V, E)$ is a $\Sigma$-graph (the nodes $q \in V$ are also called \emph{states}), $S \subseteq V$ with $S \neq \emptyset$ is the set of \emph{start states} and $T \subseteq V$ with $T \neq \emptyset$ is the set of \emph{final states}. The language $\lang(M)$ of an $\NFA$ $M$ is the set of all labels of paths from some start state to some final state. 
For a $\Sigma$-graph $G = (V, E)$, any subsets $S, T \subseteq V$ with $S \neq \emptyset \neq T$ induce the $\NFA$ $(G, S, T)$. If $S = \{s\}$ and $T = \{t\}$ are singletons, then we also write $(G, s, t)$ instead of $(G, \{s\}, \{t\})$. \par
The set $\RE_{\Sigma}$ of \emph{regular expressions} (\emph{over $\Sigma$}) is recursively defined as follows: $a \in \RE_{\Sigma}$ for every $a \in \Sigma \cup \{\eword\}$; $(\alpha \cdot \beta) \in \RE_{\Sigma}$, $(\alpha \altop \beta) \in \RE_{\Sigma}$, and $(\alpha)^+ \in \RE_{\Sigma}$, for every $\alpha,\beta\in \RE_{\Sigma}$. For any $\alpha \in \RE_{\Sigma}$, let $\lang(\alpha)$ be the regular language described by the regular expression $\alpha$ defined as usual:\footnote{As usual, we use the same notation $\lang(\cdot)$ both for regular expressions as well as finite automata.} for every $a \in \Sigma \cup \{\eword\}$, $\lang(a) = \{a\}$, and for every $\alpha, \beta \in \RE_{\Sigma}$, $\lang(\alpha \cdot \beta) = \lang(\alpha) \cdot \lang(\beta)$, $\lang(\alpha \altop \beta) = \lang(\alpha) \cup \lang(\beta)$ and $\lang(\alpha^+) = \lang(\alpha)^+$. We also use $\alpha^*$ as short hand form for $\alpha^+ \altop \eword$. By $|\alpha|$, we denote the length of $\alpha$ represented as a string.

\begin{prop}\label{constructNFAProposition}
Every regular expression $\alpha$ can be transformed in time $\bigO(|\alpha|)$ into an equivalent $\NFA$ $M = (G, p_0, p_f)$ with $|G| = \bigO(|\alpha|)$.
\end{prop}

\begin{proof}
Let $\alpha \in \RE_{\Sigma}$. We first construct the syntax tree $T_{\alpha}$ of $\alpha$ with node set $V_{T_{\alpha}}$. Obviously, $T_{\alpha}$ has size $\bigO(|\alpha|)$ and we can obtain $V_{T_{\alpha}}$ from $\alpha$ in time $\bigO(|\alpha|)$ (for example, we can transform the expression $\alpha$ to prefix notation and then construct $T_{\alpha}$ while moving through the prefix notation of $\alpha$ from left to right). We now construct an $\NFA$ $M_{\alpha} = (G = (V, E), \{p_0\}, \{p_f\})$ from $T_{\alpha}$ as follows (recall that $G$ is a $\Sigma$-graph and therefore it should adhere to our representations of $\Sigma$-graphs). We first construct an array of size $2|V_{T_{\alpha}}|$ that contains the nodes of $V = \{t_1, t_2 \mid t \in V_{T_{\alpha}}\}$ and we initialise empty $x$-adjacency lists for all these nodes and for every $x \in \Sigma \cup \{\eword\}$. Then we move through $T_{\alpha}$ top-down and if the current node $t$ is an inner node with two children $r$ and $s$, then we do the following:
\begin{itemize}
\item If $t$ corresponds to a concatenation $\cdot$, then we add $r_1$ to the $\eword$-adjacency list of $t_1$, we add $s_1$ to the $\eword$-adjacency list of $r_2$, and we add $t_2$ to the $\eword$-adjacency list of $s_2$.
\item If $t$ corresponds to an alternation $\altop$, then we add $r_1$ to the $\eword$-adjacency list of $t_1$, we add $s_1$ to the $\eword$-adjacency list of $t_1$, we add $t_2$ to the $\eword$-adjacency list of $r_2$, and we add $t_2$ to the $\eword$-adjacency list of $s_2$.
\end{itemize}
If $t$ is an inner node with one child $r$ (which means it necessarily corresponds to a $+$), then we add $r_1$ to the $\eword$-adjacency list of $t_1$, we add $t_2$ to the $\eword$-adjacency list of $r_2$, and we add $t_1$ to the $\eword$-adjacency list of $t_2$. If $t$ is a leaf labelled with $x \in \Sigma \cup \{\eword\}$, then we add $t_2$ to the $x$-adjacency list of $t_1$. Finally, if $t$ is the root of $T_{\alpha}$, we relabel $t_1$ by $p_0$ and we relabel $t_2$ by~$p_f$. \par
It can be easily verified that $\lang(M_{\alpha}) = \lang(\alpha)$. By construction $|V| = \bigO(|T_{\alpha}|) = \bigO(|\alpha|)$ and, since $G$ has constant degree, we also have that $|G| = \bigO(|V|) = \bigO(|\alpha|)$. Moreover, in the construction we spend constant time per arc that is added, so the whole construction of $M_{\alpha}$ can be done in time $\bigO(|\alpha|)$.
\end{proof}

In the following, when we speak about an \emph{automaton} (or an \emph{$\NFA$}) \emph{for a regular expression $\alpha$}, we always mean an $\NFA$ equivalent to $\alpha$ with the properties asserted by Proposition~\ref{constructNFAProposition}.\par
A $\Sigma$-graph without $\eword$-arcs is also called a \emph{graph database} (\emph{over $\Sigma$}); in the following, we denote graph databases by $\DBD = (V_{\DBD}, E_{\DBD})$. Since $V_{\DBD}$ is represented as a list, any graph database implicitly represents a linear order on $V_{\DBD}$ (i.\,e., the order induced by the list that represents $V_{\DBD}$), which we denote by $\preceq_{\DBD}$, or simply $\preceq$ if $\DBD$ is clear from the context. A class $\mathcal{C}$ of graph databases is called \emph{sparse} if there is a constant $c$ such that $|E_{\DBD}| \leq c |V_{\DBD}|$ for every $\DBD \in \mathcal{C}$. Slightly abusing notation, we shall also call single graph databases sparse to denote that we are dealing with a graph database from a sparse class of graph databases. \par
Regular expressions $q$ (over alphabet $\Sigma$) are interpreted as \emph{regular path queries} ($\RPQ$) for graph databases (over $\Sigma$). The \emph{result} $q(\DBD)$ of an $\RPQ$ $q$ on a graph database $\DBD = (V_{\DBD}, E_{\DBD})$ over $\Sigma$ is the set $q(\DBD) = \{(u, v) \mid u, v \in V_{\DBD}, \lang((\DBD, u, v)) \cap \lang(q) \neq \emptyset\}$.\par
If we interpret $q$ as a \emph{Boolean} $\RPQ$, then the result is $q_{\boole}(\DBD) = \mathsf{true}$ if $q(\DBD) \neq \emptyset$ and $q_{\boole}(\DBD) = \mathsf{false}$ otherwise. %
We consider the $\RPQ$-evaluation problems summarised in Table~\ref{ProblemsTable}.
By \emph{sorted $\enumProb$} (or \emph{semi-sorted} $\enumProb$), we denote the variant of $\enumProb$, where the pairs of $q(\DBD)$ are to be enumerated in lexicographical order with respect to $\preceq_{\DBD}$ (or ordered only with respect to their left elements, while successive pairs with the same right element can be ordered arbitrarily, respectively).

\begin{table}
\begin{tabular}{l|l|l}
Name & Input & Task\\\hline
$\booleProb$ & $\DBD$, $q$ & Decide whether $q_{\boole}(\DBD) = \mathsf{true}$. \\
$\checkProb$ & $\DBD$, $q$, $u, v$ & Decide whether $(u, v) \in q(\DBD)$. \\
$\witnessProb$ & $\DBD$, $q$ & Compute a witness $(u, v) \in q(\DBD)$ or report that none exists. \\\hline
$\evalProb$ & $\DBD$, $q$ & Compute the whole set $q(\DBD)$ \\
$\countProb$ & $\DBD$, $q$ & Compute $|q(\DBD)|$. \\\hline
$\enumProb$ & $\DBD$, $q$ & Enumerate the whole set $q(\DBD)$.\medskip
\end{tabular}%\label{ProblemsTable} % fixed multiple defined reference
\caption{The investigated $\RPQ$-evaluation problems ($\DBD$ is a graph database, $q$ an $\RPQ$ and $u, v$ two nodes from $\DBD$).}
\label{ProblemsTable}
\end{table}

\begin{rem}
If an order $\preceq'$ on $V_{\DBD}$ is explicitly given as a bijection $\pi : V_{\DBD} \to \{1, \ldots, n\}$, then we can modify $\DBD$ (in $\bigO(|V_{\DBD}|)$) such that $\preceq_{\DBD}\; = \; \preceq'$. In this regard, sorted $\enumProb$ just models the case where we wish the enumeration to be sorted according to some order. In particular, by assuming the order $\preceq_{\DBD}$ to be implicitly represented by $\DBD$, we \emph{do not} hide the complexity of sorting $n$ elements.
\end{rem}

A graph database $\DBD$ is \emph{well-formed} if $V_{\DBD} = [n]$ for some $n \in \mathbb{N}$ and $\preceq_{\DBD}$ corresponds to $\leq$ on $[n]$.

\begin{lem}\label{makeWellFormedLemma}
Let $\DBD$ be a graph database with $|V_{\DBD}| = n$. Then we can construct in time $\bigO(|\DBD|)$ a well-formed graph database $\DBD'$ and an isomorphism $\pi : [n] \to V_{\DBD}$ between $\DBD'$ and~$\DBD$. 
\end{lem}

\begin{proof}
Let $\DBD$ be a graph database with $|V_{\DBD}| = n$. We initialise an array $A$ of size $n$ that can be addressed with the elements of $V_{\DBD}$ and the entries of which can store numbers of $[n]$, and an array $B$ of size $n$ that can be addressed with the elements of $[n]$ and the entries of which can store elements of $V_{\DBD}$. Furthermore, we initialise a counter $c = 1$. Then we move through the list for $V_{\DBD}$ from left to right and for each element $u$ that we encounter, we set $A[u] = c$, $B[c] = u$ and increment $c$. Obviously, $A$ describes an isomorphism $\pi: V_{\DBD} \to [n]$ and $B$ describes $\pi^{-1}$. Moreover, this can be done in time $\bigO(|V_{\DBD}|)$.\par
Next, we make a copy $\DBD'$ of $\DBD$ but replace each $u \in V_{\DBD}$ by the number $A[u]$ (note that $\DBD$ is just a collection of lists that store elements of $V_{\DBD}$ along with pointers to lists). This can be done in time $\bigO(|\DBD|)$. Moreover, $\DBD'$ is obviously isomorphic to $\DBD$ and $B$ describes an isomorphism $[n] \to V_{\DBD}$.
\end{proof}

Lemma~\ref{makeWellFormedLemma} means that with an overhead of time $\bigO(|\DBD|)$, we can always assume that our input graph databases are well-formed; in particular, note that with the isomorphism $\pi$ ensured by the lemma, we can always translate elements from $q(\DBD')$, where $\DBD'$ is the well-formed graph database isomorphic to $\DBD$, back to the corresponding elements of $q(\DBD)$. Thus, Lemma~\ref{makeWellFormedLemma} justifies that whenever we spend at least $\bigO(|\DBD|)$ in some preprocessing, we can always assume that the input graph database is well-formed.

\subsection{General Algorithmic Framework for \texorpdfstring{$\RPQ$}{\textsf{RPQ}}-Evaluation} 

We assume the RAM-model with logarithmic word-size as our computational model. Let us next discuss our algorithmic framework for $\RPQ$-evaluation. The input to our algorithms is a graph database $\DBD = (V_{\DBD}, E_{\DBD})$ and an $\RPQ$ $q$ (and, for solving the problem $\checkProb$, also a pair $(u, v) \in V_{\DBD}$). \par
In the case of $\enumProb$, the algorithms have routines $\preprocess$ and $\enum$. Initially, $\preprocess$ performs some preliminary computations on the input or constructs some auxiliary data-structures; the performance of $\preprocess$ is measured in its running-time depending on the input size as usual (i.\,e., we treat $\preprocess$ as an individual algorithm). Then $\enum$ will produce an enumeration $(u_1, v_1), (u_2, v_2), \ldots, (u_\ell, v_\ell)$ such that $q(\DBD) = \{(u_i, v_i) \mid 1 \leq i \leq \ell\}$, no element occurs twice, and the algorithm reports when the enumeration is done. We measure the performance of $\enum$ in terms of its \emph{delay}, which describes the time that (in the worst-case) elapses between enumerating two consecutive elements, between the start of the enumeration and the first element, and between the last element and the end of the enumeration (or between start and end in case that $q(\DBD) = \emptyset$). We say that (variants of) $\enumProb$ can be solved \emph{with preprocessing $p$ and delay $d$}, where $p$ and $d$ are functions bounding the preprocessing running-time and the delay. In the case that $p = \bigO(d)$, the preprocessing complexity is absorbed by the delay; in this case, we say that (variants of) $\enumProb$ can be solved \emph{with delay $d$} and do not mention any bound on the preprocessing. \par
We also consider $\enumProb$ in the \emph{dynamic setting}, i.\,e., there is the possibility to perform \emph{update} operations on the input graph database $\DBD$, which trigger a routine $\update$. After an update and termination of the $\update$ routine, invoking $\enum$ is supposed to enumerate $q(\DBD')$, where $\DBD'$ is the updated graph database. The performance of an algorithm for $\enumProb$ is then measured in the running-times of routines $\preprocess$ (to be initially carried out only once) and $\update$, as well as the delay. We only consider the following types of individual updates: inserting a new arc between existing nodes, deleting an arc, adding a new (isolated) node, deleting an (isolated) node. In particular, deleting or adding a single non-isolated node $u$ may require a non-constant number of updates.

\section{The Product Graph Approach}\label{sec:algoPrelim}

The PG-approach has already been informally described in the introduction; for our fine-grained perspective, we need to define it in detail. Let $\DBD = (V_{\DBD}, E_{\DBD})$ be a graph database over some alphabet $\Sigma$ and let $q$ be an $\RPQ$ over~$\Sigma$. Furthermore, let $(G_q, p_{\init}, p_f)$ with $G_q = (V_{q}, E_{q})$ be an automaton for $q$. Recall that, according to Proposition~\ref{constructNFAProposition}, $G_q$ can be obtained in time $\bigO(|q|)$ and it has $\bigO(|q|)$ states and $\bigO(|q|)$ arcs. The \emph{product graph} of $\DBD$ and $G_q$ is the $\Sigma$-graph $\DBtimesRE{\DBD}{q} = (\DBtimesREV{\DBD}{q}, \DBtimesREE{\DBD}{q})$, where $\DBtimesREV{\DBD}{q} = \{(u, p) \mid u \in V_{\DBD}, p \in V_{q}\}$ and 
\begin{align*}
\DBtimesREE{\DBD}{q} =\:&\{((u, p), x, (v, p')) \mid (u, x, v) \in E_{\DBD}, (p, x, p') \in E_{q}\} \:\cup \\
&\{((u, p), \eword, (u, p')) \mid u \in V_{\DBD}, (p, \eword, p') \in E_{q}\}\,.
\end{align*}

\begin{rem}
The arc labels in $\DBtimesRE{\DBD}{q}$ are superfluous in the sense that we only need the underlying graph of $\DBtimesRE{\DBD}{q}$ (see Lemma~\ref{reduceEvalToSTTCLemma}). We define it nevertheless as $\Sigma$-graph, since then all our definitions and terminology for $\Sigma$-graphs introduced above apply as well. 
\end{rem}

\begin{lem}\label{crossProductSizeLemma}
$|\DBtimesREV{\DBD}{q}| = \bigO(|V_{\DBD}| |q|)$, $|\DBtimesREE{\DBD}{q}| = \bigO(|\DBD| |q|)$ and $\DBtimesRE{\DBD}{q}$ can be computed in time $\bigO(|\DBtimesRE{\DBD}{q}|) = \bigO(|\DBD||q|)$.
\end{lem}

\begin{proof}
We first note that $|\DBtimesREV{\DBD}{q}| \leq |V_{\DBD}| |q|$ directly follows from the definition. Moreover, the following is also immediate by definition:
\begin{align*}
&|\{((u, p), \eword, (u, p')) \mid u \in V_{\DBD}, (p, \eword, p') \in E_{q}\}| = \bigO(|V_{\DBD}||E_{q}|) = \bigO(|\DBD||q|) \text{ and }\\
&|\{((u, p), x, (v, p')) \mid (u, x, v) \in E_{\DBD}, (p, x, p') \in E_{q}\}| = 
\bigO(|E_{\DBD}||E_{q}|) = \bigO(|\DBD||q|)\,.
\end{align*}
Consequently, $|\DBtimesREE{\DBD}{q}| = \bigO(|\DBD| |q|)$.\par
For the question how $\DBtimesRE{\DBD}{q}$ can be computed, we have to keep in mind our list-based implementation of $\Sigma$-graphs (see Section~\ref{sec:mainDefs}). For every $u \in V_{\DBD}$ and every $p \in V_q$, we add $(u, p)$ to the list that stores $\DBtimesREV{\DBD}{q}$. This can be done by moving $|V_{\DBD}|$ times through the list for $V_q$. Thus, time $\bigO(|V_{\DBD}||V_q|) = \bigO(|\DBD||q|)$ is sufficient.\par
In order to construct the adjacency lists, we proceed as follows. Let $u \in V_{\DBD}$, $p \in V_q$ and $x \in \Sigma$. Then we add all $(v, p')$ to the $x$-adjacency list of $(u, p)$, where $v$ is an element of the $x$-adjacency list of $u$, and $p'$ is an element of the $x$-adjacency list of $p$. Moreover, we add all $(u, p')$ to the $\eword$-adjacency list of $(u, p)$, where $p'$ is an element of the $\eword$-adjacency list of $p$. This can be done by moving once through the lists of adjacency lists for $E_{\DBD}$ and, for each encountered element, to move through the lists of adjacency lists for $E_{q}$. Since each insertion to a list can be done in constant time, the whole procedure can be carried out in time $\bigO(|\Sigma||\DBD||q|) = \bigO(|\DBD||q|)$.
\end{proof}

The following lemma, which is an immediate consequence of the construction, shows how $\DBtimesRE{\DBD}{q}$ can be used for solving $\RPQ$-evaluation tasks (recall that $E^*$ is the reflexive-transitive closure of the underlying unlabelled graph).

\begin{lem}\label{reduceEvalToSTTCLemma}
For every $u, v \in V_{\DBD}$, $(u, v) \in q(\DBD)$ if and only if $((u, p_{\init}), (v, p_{f})) \in (\DBtimesREE{\DBD}{q})^*$.
\end{lem}

\section{Fine-grained Complexity and Conditional Lower Bounds}\label{sec:fineGrainedComplexity}

We now state several computational problems along with hypotheses regarding their complexity, which are commonly used in the framework of fine-grained complexity to obtain conditional lower bounds. We discuss some details and give background information later on.

\begin{itemize}
\item \underline{\textsf{Orthogonal Vectors} ($\OV$)}: Given sets $A, B$ each containing $n$ Boolean-vectors of dimension~$d$, check whether there are vectors $\vec{a} \in A$ and $\vec{b} \in B$ that are orthogonal.\\
$\OV$-Hypothesis: For every $\epsilon > 0$ there is no algorithm solving $\OV$ in time $\bigO(n^{2 - \epsilon} \poly(d))$. 
\item \underline{\textsf{Boolean Matrix Multiplication} ($\BMMProb$)}: Given Boolean $n \times n$ matrices $A, B$, compute $A \times B$.\\
$\combBMMProb$-Hypothesis: For every $\epsilon > 0$ there is no combinatorial algorithm that solves $\BMMProb$ in time $\bigO(n^{3 - \epsilon})$. 
\item \underline{\textsf{Sparse Boolean Matrix Multiplication} ($\SBMMProb$)}: Like $\BMMProb$, but all matrices are represented as sets $\{(i, j) \mid A[i, j] = 1\}$ of $1$-entries.\\
$\SBMMProb$-Hypothesis: There is no algorithm that solves $\SBMMProb$ in time $\bigO(m)$ (where $m$ is the total number of $1$-entries, i.\,e., $m = |\{(i, j) \mid A[i, j] = 1\}| + |\{(i, j) \mid B[i, j] = 1\}| + |\{(i, j) \mid (A \times B)[i, j] = 1\}|$). 
\item \underline{\textsf{Online Matrix-Vector Multiplication} ($\OMv$)}: Given Boolean $n \times n$-matrix $M$ and a sequence $\vec{v}^1, \vec{v}^2, \ldots, \vec{v}^n$ of $n$-dimensional Boolean vectors, compute the sequence $M \vec{v}^{1}, M \vec{v}^{2}, \ldots, M \vec{v}^{n}$, where $M \vec{v}^{i}$ is produced as output before $\vec{v}^{i+1}$ is received as input.\\
$\OMv$-Hypothesis: For every $\epsilon > 0$ there is no algorithm that solves $\OMv$ in time $\bigO(n^{3 - \epsilon})$. 
\end{itemize}

We will reduce these problems to variants of $\RPQ$ evaluation problems in such a way that algorithms with certain running-times for $\RPQ$ evaluation would break the corresponding hypotheses mentioned above. Thus, we obtain lower bounds for $\RPQ$ evaluation that are conditional to these hypotheses. In the following, we give a very brief overview of the relevance of these problems and corresponding hypotheses in fine-grained complexity.

The problem $\OV$ can be solved by brute-force in time $\bigO(n^2d)$ and the hypothesis that there is no subquadratic algorithm is well-established. It exists in slightly different variants and has been formulated in several different places in the literature (e.\,g.,~\cite{Bringmann2014, Bringmann2019, Williams2015}). The variant used here is sometimes referred to as \emph{moderate dimension} $\OV$-hypothesis in contrast to \emph{low dimension} variants, where $d$ can be assumed to be rather small in comparison to~$n$. The relevance of the $\OV$-hypothesis is due to the fact that it is implied by the Strong Exponential Time Hypothesis (SETH)~\cite{Williams2004, Williams2005}, and therefore it is a convenient tool to prove SETH lower bounds that has been applied in various contexts. \par
One of the most famous computational problems is $\BMMProb$, which, unfortunately, is a much less suitable basis for conditional lower bounds. The straightforward algorithm solves it in time $\bigO(n^3)$, but there are \emph{fast matrix multiplication} algorithms that run in time $\bigO(n^{2.373})$~\cite{Williams2012, Gall2014}. It is unclear how much further this exponent can be decreased and there is even belief that $\BMMProb$ can be solved in time $n^{2 + \smallO(1)}$ (see~\cite{Williams2012b} and \cite[Section~6]{BerkholzEtAl2020}). However, these theoretically fast algorithms cannot be considered efficient in a practical sense, which motivates the mathematically informal notion of ``\emph{combinatorial}'' algorithms (see, e.\,g.,~\cite{WilliamsWilliams2018}).\footnote{The term ``combinatorial algorithm'' is not well-defined, but intuitively such algorithms have running-times with low constants in the $\bigO$-notation, and are feasibly implementable.} So far, no truly subcubic \emph{combinatorial} $\BMMProb$-algorithm exists and it has been shown in~\cite{WilliamsWilliams2018} that $\BMMProb$ is contained in a class of problems (including other prominent examples like Triangle Finding (also mentioned below) and Context-Free Grammar Parsing) which are all equivalent in the sense that if one such problem is solvable in truly subcubic time by a combinatorial algorithm, then all of them are. Consequently, it is often possible to argue that the existence of a certain combinatorial algorithm for some problem would imply a major (and unlikely) algorithmic breakthrough with respect to $\BMMProb$, Parsing, Triangle Finding, etc. Despite the defect of relying on the vague notion of \emph{combinatorial} algorithms, this lower bound technique is a common approach in fine-grained complexity (see, e.\,g.,~\cite{WilliamsWilliams2018, HenzingerEtAl2015, AbboudWilliams2014, AbboudEtAl2018, AbboudEtAl2018_2, HenzingerEtAl2017}). Whenever we use the $\combBMMProb$-hypothesis, our reductions will always be combinatorial, which is necessary; moreover, whenever we say that a certain running time cannot be achieved unless the $\combBMMProb$-hypothesis fails, we mean, of course, that it cannot be achieved by a combinatorial algorithm.\par
In order to make $\BMMProb$ suitable as base problem for conditional lower bounds (that does not rely on \emph{combinatorial} algorithms) one can formulate the weaker (i.\,e., more plausible) hypothesis that $\BMMProb$ cannot be solved in time linear in the number of $1$-entries of the matrices (therefore called \emph{sparse} $\BMMProb$ since matrices are represented in a sparse way); see~\cite{AmossenPagh2009, YusterZwick2005}. Another approach is to require the output matrix $A \times B$ to be computed column by column, i.\,e., formulating it as the online-version $\OMv$. For $\OMv$, subcubic algorithms are not known and would yield several major algorithmic breakthroughs (see~\cite{HenzingerEtAl2015}). \par

A convenient tool to deal with $\BMMProb$ is the problem $\TriProb$: check whether a given undirected graph $G$ has a triangle. This is due to the fact that these two problems are subcubic equivalent with respect to combinatorial algorithms (see~\cite{WilliamsWilliams2018}), i.\,e., the $\combBMMProb$-hypothesis fails if and only if $\TriProb$ can be solved by a combinatorial algorithm in time $\bigO(n^{3 - \epsilon})$ for some $\epsilon > 0$. Thus, for lower bounds conditional to the $\combBMMProb$-hypothesis, we can make use of both these problems. There is also a (non-combinatorial) $\TriProb$-hypothesis that states that $\TriProb$ cannot be solved in linear time in the number of edges, but we were not able to apply it in the context of $\RPQ$-evaluation (see~\cite{AbboudWilliams2014} for different variants of $\TriProb$).

\section{Bounds for the Non-Enumeration Problem Variants}\label{sec:nonEnumBounds}

We now investigate how well the PG-approach performs with respect to the non-enumeration variants of $\RPQ$-evaluation, and we give some evidence that, in most cases, it can be considered optimal or almost optimal (subject to the algorithmic hypotheses of Section~\ref{sec:fineGrainedComplexity}).
 	
\subsection{Boolean Evaluation, Testing and Computing a Witness}\label{sec:booleTestWitness}

It is relatively straightforward to see that the problems $\checkProb$ and $\booleProb$ are equivalent and can both be reduced to $\witnessProb$. Hence, upper bounds for $\witnessProb$ and lower bounds for $\checkProb$ or $\booleProb$ automatically apply to all three problem variants, which simplifies the proofs for such bounds. We shall now formally prove this.

\begin{lem}\label{BooleToCheckReduction}
Let $(\DBD, q)$ be an $\booleProb$-instance. Then we can construct an equivalent $\checkProb$-instance $(\DBD', q', u, v)$ with $|\DBD'| = \bigO(|\DBD|)$ and $|q'| = \bigO(|q|)$ in time $\bigO(|\DBD| + |q|)$.
\end{lem}

\begin{proof}
Let $\DBD$ and $q$ be defined over $\Sigma$. We transform $\DBD$ into a graph database $\DBD'$ over $\Sigma \cup \{\#\}$ (where $\# \notin \Sigma$) by adding new nodes $u$ and $v$, and new arcs $(u, \#, x)$ and $(x, \#, v)$ for every $x \in V_{\DBD}$. Moreover, we set $q' = \# q \#$. This construction can be carried out in time $\bigO(|\DBD| + |q|)$ and we can also note that $|\DBD'| = \bigO(|\DBD|)$ and $|q'| = \bigO(|q|)$.\par
It remains to show that $q_{\boole}(\DBD) = \textsf{true}$ if and only if $(u, v) \in q'(\DBD')$. If $q_{\boole}(\DBD) = \textsf{true}$, then there is some $(u', v') \in q(\DBD)$, which means that in $\DBD$ there is a path $u', \ldots, v'$ that is labelled with a word $w \in \lang(q)$. Since there are arcs $(u, \#, u')$ and $(v', \#, v)$ in $\DBD'$, there is a path $u, u', \ldots, v', v$ in $\DBD'$ that is labelled with $\# w \# \in \lang(q')$. Therefore, $(u, v) \in q'(\DBD')$. On the other hand, if $(u, v) \in q'(\DBD')$, then there is some path $u, u', \ldots, v', v$ in $\DBD'$ that is labelled with $\# w \# \in \lang(q')$, so therefore also $(u', v') \in q(\DBD)$ and $q_{\boole}(\DBD) = \textsf{true}$ (note that $u' = v'$ and therefore $w = \eword$ is also possible).
\end{proof}

\begin{lem}\label{CheckToBooleReduction}
Let $(\DBD, q, u, v)$ be an $\checkProb$-instance. Then we can construct an equivalent $\booleProb$-instance $(\DBD', q')$ with $|\DBD'| = \bigO(|\DBD|)$ and $|q'| = \bigO(|q|)$ in time $\bigO(|\DBD| + |q|)$.
\end{lem}

\begin{proof}
Let $\DBD$ and $q$ be defined over $\Sigma$. We transform $\DBD$ into a graph database $\DBD'$ over $\Sigma \cup \{\#\}$ (where $\# \notin \Sigma$) by adding new nodes $s$ and $t$ with arcs $(s, \#, u)$ and $(v, \#, t)$, and we define $q' = \# q \#$. This construction can be carried out in time $\bigO(|\DBD| + |q|)$ and we can also note that $|\DBD'| = \bigO(|\DBD|)$ and $|q'| = \bigO(|q|)$.\par
It remains to show that $(u, v) \in q(\DBD)$ if and only if $q'_{\boole}(\DBD') = \textsf{true}$. If $(u, v) \in q(\DBD)$, then there is a path from $u$ to $v$ in $\DBD$ that is labelled with a word from $\lang(q)$. Since there are arcs $(s, \#, u)$ and $(v, \#, t)$ in $\DBD'$, there is a path in $\DBD'$ from $s$ to $t$ labelled with a word from $\lang(\# q \#)$; thus $(s, t) \in q'(\DBD')$ and therefore $q'_{\boole}(\DBD') = \textsf{true}$. On the other hand, if $q'_{\boole}(\DBD') = \textsf{true}$, then we can conclude that $(s, t) \in q'(\DBD')$, which is due to the fact that $q' = \# q \#$, $q$ does not contain any occurrence of $\#$ and the only arcs labelled with $\#$ have source $s$ and target $t$. This implies that there is a path $s, u, \ldots, v, t$ labelled with a word $\# w \#$ with $w \in \lang(q)$, which implies that there is a path from $u$ to $v$ labelled with a word from $\lang(q)$, which means that $(u, v) \in q(\DBD)$.
\end{proof}

\begin{thm}\label{CheckAndBooleEquivalentTheorem}
Let $f : \mathbb{N} \times \mathbb{N} \to \mathbb{N}$ be some polynomial function with $f(n_1, n_2) = \Omega(n_1 + n_2)$. 
\begin{enumerate}
\item\label{CheckAndBooleEquivalentTheoremPointOne} $\checkProb$ can be solved in time $f(|\DBD|, |q|)$ if and only if $\booleProb$ can be solved in time $f(|\DBD|, |q|)$.
\item\label{CheckAndBooleEquivalentTheoremPointTwo} If $\witnessProb$ can be solved in time $f(|\DBD|, |q|)$, then $\booleProb$ and $\checkProb$ can be solved in time $f(|\DBD|, |q|)$.
\end{enumerate}
\end{thm}

\begin{proof}
Lemmas~\ref{BooleToCheckReduction}~and~\ref{CheckToBooleReduction} directly imply Point~\ref{CheckAndBooleEquivalentTheoremPointOne}. Moreover, an algorithm that solves $\witnessProb$ implicitly solves $\booleProb$ as well. Together with Point~\ref{CheckAndBooleEquivalentTheoremPointOne}, this proves Point~\ref{CheckAndBooleEquivalentTheoremPointTwo}. 
\end{proof}

We are now ready to prove that the PG-approach yields the following upper bound.

\begin{thm}\label{CCCheckUpperBoundsTheorem}
The problems $\checkProb$, $\booleProb$ and $\witnessProb$ can be solved in time $\bigO(|\DBD||q|)$. 
\end{thm}

\begin{proof}
We only show the upper bound for $\witnessProb$ (due to Theorem~\ref{CheckAndBooleEquivalentTheorem} it applies to $\checkProb$ and $\booleProb$ as well). To this end, let $\DBD = (V_{\DBD}, E_{\DBD})$ be a graph database over $\Sigma$ and let $q$ be an $\RPQ$ over $\Sigma$. We construct $\DBtimesRE{\DBD}{q}$, which, according to Lemma~\ref{crossProductSizeLemma}, can be done in time $\bigO(|\DBD||q|)$. In the following considerations, we interpret $\DBtimesRE{\DBD}{q}$ as its underlying non-labelled graph.\par
We add to $\DBtimesRE{\DBD}{q}$ a node $v_{\mathsf{source}}$ with an arc to each $(u, p_0)$ with $u \in V_{\DBD}$. This can be done in time $\bigO(|V_{\DBD}||q|)$ as follows. First, we add $v_{\mathsf{source}}$ to the list that represents $\DBtimesREV{\DBD}{q}$, which requires constant time. Then, we move through the list that represents $\DBtimesREV{\DBD}{q}$ and every node $(u, p_0)$ that we encounter is added to the adjacency list for $v_{\mathsf{source}}$.\par
Next, we perform a special kind of BFS in $v_{\mathsf{source}}$. First, we construct an array $S$ of size $|\DBtimesREV{\DBD}{q}|$ the entries of which can store values from $\DBtimesREV{\DBD}{q} \cup \{0\}$ and which can be addressed by the nodes from $\DBtimesREV{\DBD}{q}$. Moreover, $S$ is initialised with every entry storing~$0$. Then, for every $u \in V_{\DBD}$, we set $S[(u, p_0)] = u$. This can be done in time $\bigO(|\DBtimesREV{\DBD}{q}|) = \bigO(|\DBD||q|)$. We perform a BFS from $v_{\textsf{source}}$ and whenever we traverse an arc $((u, p), (u', p'))$, we set $S[(u', p')] = S[(u, p)]$. This can be done in time $\bigO(|\DBtimesRE{\DBD}{q}|) = \bigO(|\DBD||q|)$. Since $S[(u, p_0)] = u$ for every $u \in V_{\DBD}$, we can conclude by induction that, for every $(v, p) \in \DBtimesREV{\DBD}{q}$, if $(v, p)$ is reachable from some node $(u, p_0)$, then $S[(v, p)] = u'$ for some $u' \in V_{\DBD}$ such that $(v, p)$ is reachable from $(u', p_0)$, and if $(v, p)$ is not reachable from any node $(u, p_0)$, then $S[(u, p)] = 0$. Consequently, if there is some $v \in V_{\DBD}$ with $S[(v, p_f)] = u \neq 0$, then there is a path from $(u, p_0)$ to $(v, p_f)$ in $\DBtimesRE{\DBD}{q}$, which, according to Lemma~\ref{reduceEvalToSTTCLemma}, means that $(u, v) \in q(\DBD)$ and therefore, we can produce the output $(u, v)$. On the other hand, if there is no $v \in V_{\DBD}$ with $S[(v, p_f)] \neq 0$, then there are no $u, v \in V_{\DBD}$ with a path from $(u, p_0)$ to $(v, p_f)$; thus, $q(\DBD) = \emptyset$. Checking whether there is some $v \in V_{\DBD}$ with $S[(v, p_f)] \neq 0$ can be done in time $\bigO(|V_{\DBD}|)$.
\end{proof}

More interestingly, we can complement this upper bound with lower bounds as follows.

\begin{thm}\label{CCCheckLowerOVBoundsTheorem}
If any of the problems $\checkProb$, $\booleProb$ and $\witnessProb$ can be solved in time $\bigO(|\DBD|^{2-\epsilon} + |q|^{2})$ or $\bigO(|\DBD|^{2} + |q|^{2-\epsilon})$ for some $\epsilon > 0$, then the $\OV$-hypothesis fails. This lower bound also holds for the restriction to sparse graph databases.
\end{thm}

\begin{proof}
We prove the lower bound for $\checkProb$ only, since by Theorem~\ref{CheckAndBooleEquivalentTheorem} it also applies to $\witnessProb$ and $\booleProb$. We first devise a general reduction from the $\OV$-problem to $\checkProb$ (which is similar to the reduction from~\cite{BackursIndyk2016} used for proving conditional lower bounds of regular expression matching):\par
Let $A = \{\vec{a}^1, \ldots, \vec{a}^n\}$ and $B = \{\vec{b}^1, \ldots, \vec{b}^n\}$ be an instance for the $\OV$-problem. We define an $\RPQ$ $q$ and a graph database $\DBD$ over the alphabet $\Sigma = \{0, 1, \#\}$ as follows. For every $i \in [n]$, let $w_i = \vec{b}^i[1]\vec{b}^i[2]\ldots \vec{b}^i[d]$, and let $q = \# (w_1 \altop w_2 \altop \ldots \altop w_n) \#$. Moreover, let $\DBD = (V_{\DBD}, E_{\DBD})$, where $V_{\DBD}$ contains nodes $s$ and $t$ and, for every $i \in [n]$, nodes $v_{i, 0}, v_{i, 1}, \ldots, v_{i, d}$. For every $i \in [n]$ and $j \in \{0\} \cup [d-1]$, there is an arc from $v_{i, j}$ to $v_{i, j+1}$ labelled with $0$ and, if $\vec{a}^i[j+1] = 0$, also an arc from $v_{i, j}$ to $v_{i, j+1}$ labelled with $1$. Finally, for every $i \in [n]$, there are arcs labelled with $\#$ from $s$ to $v_{i, 0}$ and from $v_{i, d}$ to $t$. It can be easily verified that there are orthogonal $\vec{a} \in A$ and $\vec{b} \in B$ if and only if $(s, t) \in q(\DBD)$. Moreover, $|\DBD| = \bigO(|A|d) = \bigO(nd)$ and $|q| = \bigO(|B|d) = \bigO(nd)$, and, furthermore, $\DBD$ and $q$ can also be constructed in time $\bigO(nd)$, and, since $|E_{\DBD}| = \bigO(|V_{\DBD}|)$, $\DBD$ is a sparse graph database.\par
We now assume that $\checkProb$ can be solved in time $\bigO(|\DBD|^2 + |q|^{2-\epsilon})$ for some $\epsilon > 0$. Let again $A = \{\vec{a}^1, \ldots, \vec{a}^n\}$ and $B = \{\vec{b}^1, \ldots, \vec{b}^n\}$ be an instance for the $\OV$-problem, and let $\epsilon'$ be arbitrarily chosen with $0 < \epsilon' < \epsilon$. We divide $A$ into $A_1, A_2, \ldots, A_{\lceil n^{\epsilon'} \rceil}$ with $|A_i| = \lceil n^{1-\epsilon'} \rceil$ for every $i \in [\lceil n^{\epsilon'} \rceil]$ and such that $\bigcup^{\lceil n^{\epsilon'} \rceil}_{i = 1} A_i = A$. Obviously, $(A, B)$ is a positive $\OV$-instance if and only if at least one of $(A_1, B), (A_2, B), \ldots, (A_{\lceil n^{\epsilon'} \rceil}, B)$ is a positive $\OV$-instance. We can now separately reduce each $(A_i, B)$ to an $\checkProb$-instance $(\DBD_i, q)$ as described above and we note that $|\DBD_i| = \bigO(|A_i|d) = \bigO(\lceil n^{1-\epsilon'} \rceil d)$ and $|q| = \bigO(|B|d) = \bigO(nd)$. Then we solve each instance in the assumed time bound $\bigO(|\DBD|^2 + |q|^{2-\epsilon})$, which requires total time of $\bigO(\lceil n^{\epsilon'} \rceil ((\lceil n^{1-\epsilon'} \rceil d)^2 + (nd)^{2-\epsilon}) = \bigO((n^{2-\epsilon'} + n^{2-\epsilon + \epsilon'}) \poly(d)) = \bigO(n^{2-(\epsilon - \epsilon')}) \poly(d))$, where $(\epsilon - \epsilon') > 0$. This contradicts the $\OV$-hypothesis.\par
The assumption that $\checkProb$ can be solved in time $\bigO(|\DBD|^{2 - \epsilon} + |q|^2)$ for some $\epsilon > 0$ can be handled analogously. We divide again $A$ into $A_1, A_2, \ldots, A_{\lceil n^{\epsilon'} \rceil}$ as described above, but then we reduce the $\OV$-instances $(B, A_1), (B, A_2), \ldots, (B, A_{\lceil n^{\epsilon'} \rceil})$ to the $\checkProb$-instances $(\DBD, q_1), (\DBD, q_2), \ldots, (\DBD, q_{\lceil n^{\epsilon'} \rceil})$. Note that $|\DBD| = \bigO(nd)$ and $|q_i| = \bigO(|A_i|d) = \bigO(\lceil n^{1-\epsilon'} \rceil d)$. By assumption, each of these instances can be solved in time $\bigO(|\DBD|^{2 - \epsilon} + |q_i|^2) = \bigO((nd)^{2 - \epsilon} + (\lceil n^{1-\epsilon'} \rceil d)^2)$, which again leads to a total running-time of  $\bigO(\lceil n^{\epsilon'} \rceil ((nd)^{2 - \epsilon} + (\lceil n^{1-\epsilon'} \rceil d)^2)) = \bigO(n^{2-(\epsilon - \epsilon')}) \poly(d))$, where $(\epsilon - \epsilon') > 0$. This contradicts the $\OV$-hypothesis.\par
\end{proof}

Since $(|\DBD||q|)^{1-\epsilon}\leq ((\max\{|\DBD|,|q|\})^2)^{1-\epsilon}=\max\{|\DBD|^{2-2\epsilon}, |q|^ {2-2\epsilon}\}\leq |\DBD|^{2-\epsilon}+|q|^2$, Theorem~\ref{CCCheckLowerOVBoundsTheorem} also rules out running times of the form $\bigO((|\DBD||q|)^{1-\epsilon})$ and $\bigO(\max\{|\DBD|,|q|\}^{2-\epsilon})$, but does not exclude a running time of $\bigO(|\DBD|^{2-\epsilon}+f(|q|))$ with $f(|q|) = \Omega(|q|^2)$. However, for $\epsilon < 1$ this is super-linear in $|\DBD|$ (and therefore inferior to $\bigO(|\DBD||q|)$ under the assumption $|q| \ll |\DBD|$, i.\,e., that the database is much larger than the query), and for $\epsilon = 1$, we would obtain $\bigO(|\DBD| + f(|q|))$, which, under the assumption $|q| \ll |\DBD|$, is a small and arguably negligible asymptotic improvement over $\bigO(|\DBD||q|)$.\par
If the size of $\DBD$ is expressed in terms of $|V_{\DBD}|$, then Theorem~\ref{CCCheckUpperBoundsTheorem} also gives an upper bound of $\bigO(|V_{\DBD}|^2|q|)$. In this regard, we can show the following lower bound.

\begin{thm}\label{CCCheckLowerTriangleBoundsTheorem}
If any of the problems $\checkProb$, $\booleProb$ and $\witnessProb$ can be solved in time $\bigO(|V_{\DBD}|^{3-\epsilon} + |q|^{3-\epsilon})$ for some $\epsilon > 0$, then the $\combBMMProb$-hypothesis fails. 
\end{thm}

\begin{proof}
We prove the lower bound for $\booleProb$ only, since by Theorem~\ref{CheckAndBooleEquivalentTheorem} it also applies to $\witnessProb$ and $\booleProb$. \par
We assume that $\booleProb$ can be solved by a combinatorial algorithm in time $\bigO(|V_{\DBD}|^{3-\epsilon} + |q|^{3-\epsilon})$ for some $\epsilon > 0$. We can then solve $\TriProb$ on some instance $G = (V, E)$ with $V = \{v_1, v_2, \ldots, v_n\}$ as follows (the result will then follow from the combinatorial subcubic equivalence of $\BMMProb$ and $\TriProb$ (see Section~\ref{sec:fineGrainedComplexity})). We construct the graph database $\DBD$ over $\{\ta, \#\}$ with nodes 
\[
V_{\DBD} = \{s', t'\} \cup \{ s_j, t_j, \mid 1 \leq j \leq n\} \cup \{(u, i) \mid 0 \leq i \leq 3, u \in V\}
\]
and arcs 
\begin{align*}
E_{\DBD} = &\{((u, i), \ta, (u', i+1)) \mid 0 \leq i \leq 2, (u, u') \in E\}\:\cup \\
&\{(s', \#, s_1), (s_i, \ta, s_{i + 1}), (t_i, \ta, t_{i + 1}), (t_n, \#, t') \mid 1 \leq i \leq n-1\}\:\cup \\
&\{(s_i, \ta, (v_i, 0)), ((v_i, 3), \ta, t_i) \mid 1 \leq i \leq n\}\,.
\end{align*}
Furthermore, we define the $\RPQ$ $q = \# \ta^{n + 4} \#$. \par
By definition, $q(\DBD) = \{(s', t')\}$ or $q(\DBD) = \emptyset$. We call any path $p$ from $s'$ to $t'$ an \emph{$i$-$j$-path} for some $i, j \in [n]$ if there are $\ell_1, \ell_2 \in [n]$ such that 
\[
p = (s', s_1, \ldots, s_i, v_{i,0}, v_{\ell_1,1}, v_{\ell_2,2}, v_{j,3}, t_j, \ldots, t_n, t')\,.
\] 
It can be easily seen that any path from $s'$ to $t'$ is an $i$-$j$-path for some $i, j \in [n]$, and that any $i$-$j$-path is labelled with $\# \ta^i \ta^3 \ta^{n-j+1} \# = \# \ta^{i + n - j + 4} \#$. Hence, $q(\DBD) \neq \emptyset$ if and only if there is an $i$-$i$-path for some $i \in [n]$. Since, for every $i \in [n]$, there is an $i$-$i$-path if and only if there is a path $(v_{i,0}, v_{\ell_1,1}, v_{\ell_2,2}, v_{i,3})$, we see that, for every $i \in [n]$, there is an $i$-$i$-path if and only if $G$ has triangle that contains $v_i$. Consequently, $q(\DBD) \neq \emptyset$ if and only if $G$ has a triangle. \par
This means that by first constructing $\DBD$ and $q$ in time $\bigO(|D| + |q|) = \bigO(|G| + |V|)$ and then solving $\booleProb$ for instance $(\DBD, q)$ in the assumed time $\bigO(|V_{\DBD}|^{3-\epsilon} + |q|^{3-\epsilon})$ for some $\epsilon > 0$, means that we can solve $\TriProb$ by a combinatorial algorithm in time $\bigO(|V_{\DBD}|^{3-\epsilon} + |q|^{3-\epsilon}) = \bigO(|V|^{3-\epsilon})$. \par
With Lemma~\ref{CheckAndBooleEquivalentTheorem}, we conclude that the assumptions that $\checkProb$ or $\witnessProb$ can be solved in time $\bigO(|V_{\DBD}|^{3-\epsilon} + |q|^{3-\epsilon})$ for some $\epsilon > 0$, leads to a combinatorial $\bigO(|V|^{3-\epsilon})$ algorithm for $\TriProb$ as well. Thus, due to the combinatorial subcubic equivalence of $\BMMProb$ and $\TriProb$ (see Section~\ref{sec:fineGrainedComplexity}), there is a combinatorial $\bigO(n^{3-\epsilon})$ algorithm for $\BMMProb$, contradicting the $\combBMMProb$-hypothesis.
\end{proof}

Since $\bigO((|\DBD||q|)^{1 - \epsilon}) \subseteq \bigO((|V_{\DBD}|^2|q|)^{1 - \epsilon}) \subseteq \bigO(|V_{\DBD}|^{3-\epsilon} + |q|^{3-\epsilon})$, such running-times are also ruled out under the $\combBMMProb$-hypothesis. Especially, a combinatorial algorithm with running-time $\bigO((|\DBD||q|)^{1 - \epsilon})$ refutes \emph{both} the $\OV$- and the $\combBMMProb$-hypothesis; thus, such an algorithm does not exist provided that at least one of these hypotheses is true (basing lower bounds on several hypotheses is common in fine-grained complexity, see, e.\,g.,~\cite{AbboudEtAl2018}).\par
The lower bounds discussed above are only meaningful for combined complexity. However, the upper bound of Theorem~\ref{CCCheckUpperBoundsTheorem} already yields the optimum of \emph{linear} data complexity.

\subsection{Full Evaluation and Counting}

The following upper bound is again a straightforward application of the PG-approach.

\begin{thm}\label{crossProductUpperBoundEvalTheorem}
$\evalProb$ can be solved in time $\bigO(|V_{\DBD}||\DBD||q|)$. 
\end{thm}

\begin{proof}
Let $\DBD = (V_{\DBD}, E_{\DBD})$ be a graph database over $\Sigma$ and let $q$ be an $\RPQ$ over $\Sigma$. We construct $\DBtimesRE{\DBD}{q}$, which, according to Lemma~\ref{crossProductSizeLemma}, can be done in time $\bigO(|\DBD||q|)$. We interpret $\DBtimesRE{\DBD}{q}$ as its underlying non-labelled graph. According to Lemma~\ref{reduceEvalToSTTCLemma}, we can now compute $q(\DBD)$ by performing a BFS from each node $(u, p_0)$ with $u \in V_{\DBD}$. Each such BFS requires time $\bigO(|\DBtimesRE{\DBD}{q}|) = \bigO(|\DBD||q|)$, which means that the total running-time is $\bigO(|V_{\DBD}||\DBtimesRE{\DBD}{q}|) = \bigO(|V_{\DBD}||\DBD||q|)$.
\end{proof}

Instead of using graph-searching techniques on $\DBtimesRE{\DBD}{q}$, we could also compute the complete transitive closure of $\DBtimesRE{\DBD}{q}$ with fast matrix multiplication.

\begin{thm}\label{BMMUpperBoundEvalTheorem}
If $\BMMProb$ can be solved in time $\bigO(n^{\omega})$ with $\omega \geq 2$, then $\evalProb$ can be solved in time $\bigO(|V_{\DBD}|^{\omega}|q|^{\omega})$. 
\end{thm}

\begin{proof}
We assume that $\BMMProb$ can be solved in time $\bigO(n^{\omega})$. Let $\DBD = (V_{\DBD}, E_{\DBD})$ be a graph database over $\Sigma$ and let $q$ be an $\RPQ$ over $\Sigma$. We construct $\DBtimesRE{\DBD}{q}$, which, according to Lemma~\ref{crossProductSizeLemma}, can be done in time $\bigO(|\DBD||q|)$, and we interpret it as its underlying non-labelled graph. Obviously, we can turn the adjacency list-based representation of $\DBtimesRE{\DBD}{q}$ into an adjacency matrix-based representation in time $\bigO(|\DBtimesREV{\DBD}{q}|^2) = \bigO((|V_{\DBD}||q|)^2) = \bigO((|V_{\DBD}||q|)^{\omega})$. Then, we compute the transitive closure $(\DBtimesREE{\DBD}{q})^*$ of $\DBtimesRE{\DBD}{q}$, which can be done in time $\bigO(|\DBtimesREV{\DBD}{q}|^{\omega}) = \bigO((|V_{\DBD}||q|)^{\omega})$ (see~\cite{Munro1971}). In order to obtain $q(\DBD)$, it is sufficient to go through all elements $((u, p), (v, p')) \in (\DBtimesREE{\DBD}{q})^*$ and add $(u, v)$ to a new set if and only if $p = p_0$ and $p' = p_f$ (see Lemma~\ref{reduceEvalToSTTCLemma}). Since $|(\DBtimesREE{\DBD}{q})^*| \leq |\DBtimesREV{\DBD}{q}|^2 = \bigO(|\DBtimesREV{\DBD}{q}|^{\omega}) = \bigO((|V_{\DBD}||q|)^{\omega})$, this can be done in time $\bigO((|V_{\DBD}||q|)^{\omega})$.
\end{proof}

We mention this theoretical upper bound for completeness, but stress the fact that our main interest lies in \emph{combinatorial} algorithms. In addition to the limitation that algorithms for fast matrix multiplication are not practical, we also observe that the approach of Theorem~\ref{BMMUpperBoundEvalTheorem} is only better if the graph database is not too sparse, i.\,e., only if $|V_{\DBD}||\DBD| = \Omega(|V_{\DBD}|^{\omega})$.\par
Next, we investigate the question whether $\bigO(|V_{\DBD}||\DBD||q|)$ is optimal for $\evalProb$ at least with respect to combinatorial algorithms. Since for $\evalProb$ the PG-approach does not yield an algorithm that is linear in data complexity (like it was the case with respect to the problems of Section~\ref{sec:booleTestWitness}), the question arises whether the $|V_{\DBD}||\DBD|$ part can be improved at the cost of spending more time in $|q|$. It seems necessary that respective \emph{data complexity lower bounds} need reductions that do not use $q$ to represent a non-constant part of the instance, as it was the case for both the $\OV$ and the $\TriProb$ reduction from Section~\ref{sec:booleTestWitness}. \par
It is not difficult to see that multiplying two $n \times n$ Boolean matrices reduces to $\evalProb$ as formally stated by the next lemma.

\begin{lem}\label{BMMLowerBoundEvalLemma}
If $\evalProb$ can be solved in time $\bigO(|V_{\DBD}|^{\omega} f(|q|))$ for some function $f$ and $\omega \geq 2$, then $\BMMProb$ can be solved in time $\bigO(n^{\omega})$. 
\end{lem}

\begin{proof}
We assume that $\evalProb$ can be solved in time $\bigO(|V_{\DBD}|^{\omega} f(|q|))$ for some function~$f$. Let $A$ and $B$ be $n \times n$ Boolean matrices. Then we  construct the graph database $\DBD_{A, B}$ over $\{\ta\}$ with $V_{\DBD_{A, B}} = \{(i, 0), (i, 1), (i, 2) \mid i \in [n]\}$ and $E_{\DBD_{A, B}} = \{((i, 0), \ta, (j, 1)) \mid A[i, j] = 1\} \cup \{((i, 1), \ta, (j, 2)) \mid B[i, j] = 1\}$, and the $\RPQ$ $q = \ta \ta$. Obviously, $q(\DBD_{A, B}) \subseteq \{((i, 0), (j, 2)) \mid i, j \in [n]\}$ and $((i, 0), (j, 2)) \in q(\DBD_{A, B})$ if and only if $(A \times B)[i, j] = 1$. Hence, we can construct $A \times B$ from $q(\DBD_{A, B})$ in time $\bigO(|q(\DBD_{A, B})|) = \bigO(n^2)$. By assumption, the set $q(\DBD)$ can be computed in time $\bigO(|V_{\DBD_{A, B}}|^{\omega} f(|q|)) = \bigO(n^{\omega} f(2)) = \bigO(n^{\omega})$.
\end{proof}

From this, we can conclude the following lower bound.

\begin{thm}\label{BMMLowerBoundEvalTheorem}
If $\evalProb$ can be solved in time $\bigO((|V_{\DBD}||\DBD|)^{1 - \epsilon}f(|q|))$ for some function~$f$ and some $\epsilon > 0$, then the $\combBMMProb$-hypothesis fails.
\end{thm}

\begin{proof}
If $\evalProb$ can be solved by a combinatorial algorithm with a running time in $\bigO((|V_{\DBD}||\DBD|)^{1 - \epsilon}f(|q|))$ for some function~$f$ and some $\epsilon > 0$, then, since we have that $\bigO((|V_{\DBD}||\DBD|)^{1 - \epsilon}f(|q|)) = \bigO(|V_{\DBD}|^{3 - \epsilon}f(|q|))$, Lemma~\ref{BMMLowerBoundEvalLemma} implies that $\BMMProb$ can be solved in time $\bigO(n^{3 - \epsilon})$. Thus, the $\combBMMProb$-hypothesis fails.
\end{proof}

If we drop the restriction to combinatorial algorithms, we can nevertheless show (with more or less the same construction) that linear time in data complexity is impossible, unless the $\SBMMProb$-hypothesis fails. However, since the size of the output $q(\DBD)$ might be super-linear in $|\DBD|$, we should interpret \emph{linear} as linear in $|\DBD| + |q(\DBD)|$.

\begin{lem}\label{SBMMLowerBoundEvalLemma}
If $\evalProb$ can be solved in time $\bigO((|q(\DBD)| + |\DBD|)^{\omega} f(|q|))$ for some function~$f$ and $\omega \geq 1$, then $\SBMMProb$ can be solved in time $\bigO(m^{\omega})$. 
\end{lem}

\begin{proof}
We assume that $\evalProb$ can be solved in time $\bigO((|q(\DBD)| + |\DBD|)^{\omega} f(|q|))$ for some function $f$. Let $A$ and $B$ be $n \times n$ Boolean matrices given as sets of their $1$-entries and let $m$ be the total number of $1$-entries in $A$, $B$ and $A \times B$. Then we can construct a graph database $\DBD_{A, B}$ over $\{\ta\}$ as follows. For every $i \in [n]$, the set $V_{\DBD_{A, B}}$ contains a node $(i, 0)$ if the $i^{\text{th}}$ row of $A$ contains at least one $1$-entry, a node $(i, 1)$ if the $i^{\text{th}}$ column of $A$ contains at least one $1$-entry or the $i^{\text{th}}$ row of $B$ contains at least one $1$-entry, and 
a node $(i, 2)$ if the $i^{\text{th}}$ column of $B$ contains at least one $1$-entry. The set of arcs is defined by $E_{\DBD_{A, B}} = \{((i, 0), \ta, (j, 1)) \mid A[i, j] = 1\} \cup \{((i, 1), \ta, (j, 2)) \mid B[i, j] = 1\}$. We observe that $|E_{\DBD_{A, B}}| = \bigO(m)$ and, since every node in $V_{\DBD_{A, B}}$ has degree at least~$1$, $|V_{\DBD_{A, B}}| = \bigO(|E_{\DBD_{A, B}}|)$. Consequently, $|\DBD_{A, B}| = \bigO(m)$ and $\DBD_{A, B}$ can be constructed in time $\bigO(m)$. Further, we define $q=\ta \ta$.\par
We can observe that $q(\DBD_{A, B}) \subseteq \{((i, 0), (j, 2)) \mid i, j \in [n]\}$ and, for every $i, j \in [n]$, $((i, 0), (j, 2)) \in q(\DBD_{A, B})$ if and only if $(A \times B)[i, j] = 1$. Thus, $|q(\DBD_{A, B})| = \bigO(m)$ and we can obtain a set of exactly the $1$-entries of $A \times B$ from $q(\DBD_{A, B})$ in time $\bigO(|q(\DBD_{A, B})|) = \bigO(m)$. \par
By assumption, we can compute $q(\DBD_{A, B})$ in time $\bigO((|q(\DBD_{A, B})| + |\DBD_{A, B}|)^{\omega} f(|q|))$ for some function $f$ and $\omega \geq 1$. This implies that we can obtain a set of exactly the $1$-entries of $A \times B$ in time $\bigO((|q(\DBD_{A, B})| + |\DBD_{A, B}|)^{\omega} f(|q|)) = \bigO(m^{\omega} f(2)) = \bigO(m^{\omega})$.
\end{proof}

Lemma~\ref{SBMMLowerBoundEvalLemma} directly implies the following lower bound.

\begin{thm}\label{SBMMLowerBoundEvalTheorem}
If $\evalProb$ can be solved in time $\bigO((|q(\DBD)| + |\DBD|) f(|q|))$ for some function $f$, then the $\SBMMProb$-hypothesis fails.
\end{thm}

Surprisingly, we can obtain a more complete picture for the problem $\countProb$. First, we observe that obviously all upper bounds carry over from $\evalProb$ to $\countProb$. On the other hand, a combinatorial $\bigO((|V_{\DBD}||\DBD|)^{1-\epsilon}f(q))$ algorithm or a general $\bigO((|q(\DBD)| + |\DBD|) f(|q|))$ algorithm for $\countProb$ does not seem to help for solving Boolean matrix multiplication (and therefore, the lower bounds do not carry over). Fortunately, it turns out that $\OV$ is a suitable problem to reduce to $\countProb$, although by a rather different reduction compared to the one used for Theorem~\ref{CCCheckLowerOVBoundsTheorem}.

\begin{thm}\label{OVLowerBoundCountTheorem}
If $\countProb$ can be solved in time $\bigO(|\DBD|^{2-\epsilon} f(|q|))$ for some function~$f$ and $\epsilon > 0$, then the $\OV$-hypothesis fails.
\end{thm}

\begin{proof}
We assume that $\countProb$ can be solved in time $\bigO(|\DBD|^{2-\epsilon} f(|q|))$ for some function $f$ and $\epsilon > 0$. Let $A = \{\vec{a}_1, \vec{a}_2, \ldots, \vec{a}_n\}$ and $B = \{\vec{b}_1, \vec{b}_2, \ldots, \vec{b}_n\}$ be an $\OV$-instance, i.\,e., for every $i \in [n]$, $\vec{a}_i$ and $\vec{b}_i$ are $d$-dimensional Boolean vectors. Now let $A'$ be the Boolean matrix having rows $\vec{a}_1, \vec{a}_2, \ldots, \vec{a}_n$ and let $B'$ be the Boolean matrix having columns $\vec{b}_1, \vec{b}_2, \ldots, \vec{b}_n$. It can be easily seen that, for every $i, j \in [n]$, $(A' \times B')[i, j] = 0$ if and only if $\vec{a}_i$ and $\vec{b}_j$ are orthogonal. Moreover, $A'$ and $B'$ can be constructed in time $\bigO(nd)$.\par
Next, we construct $\DBD_{A', B'}$ as in the proof of Lemma~\ref{BMMLowerBoundEvalLemma}, and we note that $V_{\DBD_{A', B'}} = \bigO(n + d)$ and $E_{\DBD_{A', B'}} = \bigO(nd)$; in particular, we construct $\DBD_{A', B'}$ in time $\bigO(nd)$. Since, for every $i, j \in [n]$, $((i, 0), (j, 2)) \in q(\DBD_{A, B})$ if and only if $(A \times B)[i, j] = 1$, we know that, for every $i, j \in [n]$, $((i, 0), (j, 2)) \in q(\DBD_{A, B})$ if and only if $\vec{a}_i$ and $\vec{b}_j$ are not orthogonal. Furthermore, since $q(\DBD_{A, B}) \subseteq \{((i, 0), (j, 2)) \mid i, j \in [n]\}$, $|q(\DBD_{A, B})| = n^2$ if and only if there are no $\vec{a} \in A$ and $\vec{b} \in B$ that are orthogonal. Consequently, we can check whether there are no $\vec{a} \in A$ and $\vec{b} \in B$ that are orthogonal by computing $|q(\DBD_{A, B})|$, which, by assumption, can be done in time $\bigO(|\DBD_{A, B}|^{2-\epsilon} f(|q|)) = \bigO((nd)^{2-\epsilon} f(2)) = \bigO(n^{2-\epsilon} \poly(d))$.
\end{proof}

Since Theorem~\ref{OVLowerBoundCountTheorem} also excludes running time $\bigO((|V_{\DBD}| |\DBD|)^{1-\epsilon} f(|q|))$ for any function $f$ and $\epsilon > 0$ (without restriction to combinatorial algorithms), it also shows that, subject to the $\OV$-hypothesis, $\bigO(|V_{\DBD}||\DBD|)$ is a tight bound for the data complexity of $\countProb$.

\section{Bounds for the Enumeration of $\RPQ$s}\label{sec:enumBounds}

By using the PG-approach for enumeration, we can obtain the following upper bound.

\begin{thm}\label{crossProductEnumTheorem}
Sorted $\enumProb$ can be solved with $\bigO(|\DBD||q|)$ delay and $\bigO(1)$ updates.
\end{thm}

\begin{proof}
Let $\DBD$ be a graph database over $\Sigma$, let $q$ be an $\RPQ$ over $\Sigma$ and let $\preceq$ be the order on $V_{\DBD}$. We assume that $V_{\DBD} = [n]$ with $1 \preceq 2 \preceq \ldots \preceq n$ (see Lemma~\ref{makeWellFormedLemma}). \medskip\\
\noindent\textbf{Preprocessing:} 
\begin{enumerate}
\item We compute $\DBtimesRE{\DBD}{q}$, which, according to Lemma~\ref{crossProductSizeLemma}, can be done in time $\bigO(|\DBD||q|)$, and we interpret $\DBtimesRE{\DBD}{q}$ as its underlying non-labelled graph. 
\item We construct two arrays $S$ and $T$ of size $|V_{\DBD}|$ such that, for every $i \in [n]$, $S[i] = (i, p_0)$ and $T[i] = (i, p_f)$. Note that this also means that pointers to the corresponding adjacency lists are stored along with the nodes in $S$ and $T$. Computing $S$ and $T$ can be done in time $\bigO(|V_{\DBD}||q|)$ as follows. We move through the list that represents $\DBtimesREV{\DBD}{q}$ and for every node $(i, p_0)$, we set $S[i] = (i, p_0)$, and for every node $(i, p_f)$, we set $T[i] = (i, p_f)$. 
\item We modify $\DBtimesRE{\DBD}{q}$ by adding a new node $v_{\mathsf{sink}}$ with an arc from each $(i, p_f)$ with $i \in [n]$. This can be done in time $\bigO(|V_{\DBD}|)$ as follows. We first add $v_{\mathsf{sink}}$ to the list that represents $\DBtimesREV{\DBD}{q}$, which requires constant time. Then, for every node $(i, p_f)$ of $T$, we add $v_{\mathsf{sink}}$ to the adjacency list for $(i, p_f)$. Again, this can be done in time $\bigO(|V_{\DBD}|)$.
\item We obtain $(\DBtimesRE{\DBD}{q})^R$, which, by Lemma~\ref{reversalLemma}, can be done in time $\bigO(|\DBtimesRE{\DBD}{q}|) = \bigO(|\DBD||q|)$. 
Then we initialise a Boolean array $S'$ of size $|V_{\DBD}|$ with entries $0$ everywhere.  We then perform a BFS in $(\DBtimesRE{\DBD}{q})^R$ starting in $v_{\mathsf{sink}}$ and, for every visited node $(i, p_0)$, we set $S'[i] = 1$. For every $i \in V_{\DBD}$, $S'[i] = 1$ if and only if $v_{\mathsf{sink}}$ and therefore some node from $T$ can be reached from $(i, p_0)$. This step can be done in time $\bigO(|\DBD||q|)$. 
\item We initialise a Boolean array $T'$ of size $|V_{\DBD}|$ with entries $0$ everywhere, which can be done in time $\bigO(|V_{\DBD}|)$. 
\end{enumerate}
\noindent\textbf{Enumeration:}
In the enumeration phase, we carry out the following procedure. 
\begin{itemize}
\item For every $i = 1, 2, \ldots, n$:
\begin{itemize}
\item If $S'[i] = 1$, then 
\begin{itemize}
\item perform a BFS in $(i, p_0)$ and for every $(j, p_f)$ that we visit, we set $T'[j] = 1$,
\item for every $j = 1, 2, \ldots, n$: produce $(i, j)$ as output if $T'[j] = 1$,
\item Set all entries of $T'$ to $0$.
\end{itemize}
\end{itemize}
\end{itemize}
\noindent\textbf{Correctness:} In the enumeration phase we perform a BFS from each $(i, p_0)$ that can reach at least one node $(j, p_f)$ (i.\,e., from each $(i, p_0)$ with $S'[i] = 1$) and after termination of this BFS, we output exactly the pairs $(i, j)$ for which $(j, p_f)$ is visited in this BFS. This directly shows that we correctly enumerate $q(\DBD)$ without repetitions. Furthermore, since we consider the start vertices for the BFSs in increasing order with respect to $\preceq$, and since after termination of a BFS started in $(i, p_0)$ we output the pairs $(i, j)$ in increasing order by the second element, the enumeration is sorted by lexicographic order. \par
Next, we estimate the delay of the enumeration. Each iteration of the main loop with $S'[i] = 1$ requires time $\bigO(|\DBtimesRE{\DBD}{q}|)$ for performing the BFS, $\bigO(|V_{\DBD}|)$ for producing the output and $\bigO(|V_{\DBD}|)$ for resetting $T'$. Thus, the total running-time is $\bigO(|\DBtimesRE{\DBD}{q}|) = \bigO(|\DBD||q|)$ in the case that $S'[i] = 1$. If, on the other hand, $S'[i] = 0$, then the iteration terminates after constant time. Moreover, since $S'[i] = 1$ implies that $i$ can reach $v_{\mathsf{sink}}$ and since $v_{\mathsf{sink}}$ can only be reached via some $(j, p_f)$, we produce at least one output in such an iteration. Consequently, the delay between two outputs is bounded by the total running time for one iteration of the main loop, which is $\bigO(|\DBD||q|)$. \medskip\\
\noindent\textbf{Updates:} If $\DBD$ is changed to $\DBD'$ by an update, then we can again perform the whole preprocessing (with respect to $\DBD'$) followed by the enumeration procedure. Technically, the preprocessing is done as a first step of the enumeration procedure (since our algorithmic framework does not allow to re-run the preprocessing after an update), which is possible since its running time of $\bigO(|\DBD'||q|)$ is completely subsumed by the time available for the first delay. 
\end{proof}

This enumeration algorithm is easy to implement and has some nice features like linear preprocessing (in data complexity), sorted enumeration and constant updates. Unfortunately, these features come more or less for free with the disappointing delay bound. Is the PG-approach therefore the wrong tool for $\RPQ$-enumeration? Or can we give evidence that linear delay is a barrier we cannot break? The rest of this work is devoted to this question.\par

Since running an algorithm for $\enumProb$ until we get the first element yields an algorithm for $\booleProb$ (with preprocessing plus delay as running time), and since running such an algorithm completely solves $\evalProb$ in time preprocessing plus $|q(\DBD)|$ times delay, we can inherit several lower bounds directly from Section~\ref{sec:nonEnumBounds}.

\begin{thm}\label{simpleEnumLowerBounds}
If, for some function $f$ and $\epsilon > 0$, $\enumProb$ can be solved with 
\begin{enumerate}
\item delay $\bigO(|\DBD|^{2-\epsilon} + |q|^{2})$ or $\bigO(|\DBD|^{2} + |q|^{2-\epsilon})$, then the $\OV$-hypothesis fails.
\item delay $\bigO(|V_{\DBD}|^{3-\epsilon} + |q|^{3-\epsilon})$, then the $\combBMMProb$-hypothesis fails. 
\item preprocessing $\bigO(|\DBD| f(|q|))$ and delay $\bigO(f(|q|))$, then the $\SBMMProb$-hypothesis fails.
\item preprocessing $\bigO(|V_{\DBD}|^{3 - \epsilon} f(|q|))$ and delay $\bigO(|V_{\DBD}|^{1 - \epsilon} f(|q|))$, then the $\combBMMProb$-hypothesis fails. 
\end{enumerate}
\end{thm}

\begin{proof}
We prove the four lower bounds of Theorem~\ref{simpleEnumLowerBounds} separately. 
\begin{enumerate}
\item Let us assume that $\enumProb$ can be solved with delay $\bigO(|\DBD|^{2-\epsilon} + |q|^{2})$ (or $\bigO(|\DBD|^{2} + |q|^{2-\epsilon})$) for some $\epsilon > 0$. Then we can solve $\booleProb$ on an instance $(\DBD, q)$ as follows. We run the assumed enumeration algorithm for $\enumProb$ on $(\DBD, q)$, but stop as soon as the first pair (if any) is produced. By assumption, this requires time $\bigO(|\DBD|^{2-\epsilon} + |q|^{2})$ (or $\bigO(|\DBD|^{2} + |q|^{2-\epsilon})$). Moreover, $(\DBD, q)$ is a positive $\booleProb$ instance if and only if the interrupted enumeration algorithm produces a pair. Thus, we can solve $\booleProb$ in time $\bigO(|\DBD|^{2-\epsilon} + |q|^{2})$ (or $\bigO(|\DBD|^{2} + |q|^{2-\epsilon})$) for some $\epsilon > 0$, which, according to Theorem~\ref{CCCheckLowerOVBoundsTheorem} means that the $\OV$-hypothesis fails.
\item Let us assume that $\enumProb$ can be solved with delay $\bigO(|V_{\DBD}|^{3-\epsilon} + |q|^{3-\epsilon})$ for some $\epsilon > 0$. Then we can solve $\booleProb$ on an instance $(\DBD, q)$ as follows. We run the assumed enumeration algorithm for $\enumProb$ on $(\DBD, q)$, but stop as soon as the first pair (if any) is produced. By assumption, this requires time $\bigO(|V_{\DBD}|^{3-\epsilon} + |q|^{3-\epsilon})$. Moreover, $(\DBD, q)$ is a positive $\booleProb$ instance if and only if the interrupted enumeration algorithm produces a pair. Thus, we can solve $\booleProb$ in time $\bigO(|V_{\DBD}|^{3-\epsilon} + |q|^{3-\epsilon})$ for some $\epsilon > 0$, which, according to Theorem~\ref{CCCheckLowerTriangleBoundsTheorem} means that the $\combBMMProb$-hypothesis fails.
\item We assume that $\enumProb$ can be solved with preprocessing $\bigO(|\DBD| f(|q|))$ and delay $\bigO(f(|q|))$ for some function $f$. By completely running the corresponding enumeration algorithm, we can solve $\evalProb$ in time $\bigO(|\DBD| f(|q|)) + \bigO(|q(\DBD)| f(|q|)) = \bigO((|q(\DBD)| + |\DBD|) f(|q|))$. According to Theorem~\ref{SBMMLowerBoundEvalTheorem} this means that $\SBMMProb$-hypothesis fails.
%$\SBMMProb$ can be solved in time $\bigO(m)$. 
\item Assume that $\enumProb$ can be solved with preprocessing $\bigO(|V_{\DBD}|^{3 - \epsilon} f(|q|))$ and delay $\bigO(|V_{\DBD}|^{1 - \epsilon} f(|q|))$ for some function $f$ and $\epsilon > 0$. By completely running the corresponding enumeration algorithm, we can solve $\evalProb$ in time $\bigO(|V_{\DBD}|^{3 - \epsilon} f(|q|) + |q(\DBD)| |V_{\DBD}|^{1 - \epsilon} f(|q|)) = \bigO(|V_{\DBD}|^{3 - \epsilon} f(|q|)) $. According to Lemma~\ref{BMMLowerBoundEvalLemma}, this means that $\BMMProb$ can be solved in time $\bigO(n^{3 - \epsilon})$.\qedhere
\end{enumerate}
\end{proof}

The first two bounds only tell us that we might not be able to lower the delay $\bigO(|\DBD||q|)$ in terms of combined complexity. While this point of view was justified for the problems discussed in Section~\ref{sec:nonEnumBounds}, it does not say anything regarding delays of the form $\bigO(|\DBD|^{1-\epsilon}f(|q|))$. The third bound, conditional to the $\SBMMProb$-hypothesis, is much more relevant, since it suggests that the optimum of linear preprocessing and constant delay is not reachable. For combinatorial algorithms, the fourth bound at least answers our main question with $|\DBD|$ replaced by $|V_{\DBD}|$: with linear preprocessing, we cannot get below a delay of $\bigO(|V_{\DBD}|)$. \par
These lower bounds can be improved significantly, if we also want to handle updates (within some reasonable time bounds).

\begin{thm}\label{enumOMvTriangleLowerBoundTheorem}
If $\enumProb$ can be solved with 
\begin{enumerate}
\item arbitrary preprocessing, $\bigO(|V_{\DBD}|^{1-\epsilon}f(|q|))$ updates and $\bigO(|V_{\DBD}|^{1-\epsilon}f(|q|))$ delay for some function $f$ and $\epsilon > 0$, then the $\OMv$-hypothesis fails.\item $\bigO(|V_{\DBD}|^{3-\epsilon}f(|q|))$ preprocessing, $\bigO(|V_{\DBD}|^{2-\epsilon}f(|q|))$ updates and $\bigO(|V_{\DBD}|^{2-\epsilon}f(|q|))$ delay for some function $f$ and $\epsilon > 0$, then the $\combBMMProb$-hypothesis fails. 
\end{enumerate}
\end{thm}

\begin{proof}
We prove the two parts of Theorem~\ref{enumOMvTriangleLowerBoundTheorem} separately.
\begin{enumerate}
\item Assume that there is an algorithm that solves $\enumProb$ with preprocessing $t_p$, delay $t_d$ and update time $t_u$. We can then solve $\OMv$ on some instance $M$ and $\vec{v}^1, \vec{v}^2, \ldots, \vec{v}^n$ as follows. We initialise the empty graph database (over alphabet $\{\ta\}$) and the $\RPQ$ $q = \ta\ta$, and then we perform the preprocessing in constant time. Then, by performing $\bigO(n^2)$ updates, we transform the current graph database into the graph database $\DBD_{M, \vec{v}^1}$ over $\{\ta\}$ with nodes $\{u_i, v_i, w \mid 1 \leq i \leq n\}$ and arcs $\{(u_i, \ta, v_j) \mid M[i, j] = 1\} \cup \{(v_j, \ta, w) \mid \vec{v}_1[j] = 1\}$, which requires time $\bigO(n^2 t_u)$. Since $q(\DBD_{M, \vec{v}^1}) = \{(u_i, w) \mid (M \vec{v}^1)[i]=1\}$, we can now compute $M \vec{v}^1$ by enumerating $q(\DBD_{M, \vec{v}^1})$ in time $\bigO(n t_d)$. In the same way, for every $i$ with $2 \leq i \leq n$, we can obtain $M \vec{v}^i$ by enumerating $q(\DBD_{M, \vec{v}^i})$, where every $\DBD_{M, \vec{v}^{i+1}}$ can be obtained from $\DBD_{M, \vec{v}^{i}}$ by $n$ updates in time $\bigO(n t_u)$. Consequently, the whole procedure requires time $\bigO(n^2 t_u + n t_d + (n-1) (n t_u + n t_d))$, and the number of nodes always satisfies $|V_{\DBD}|\in\bigO(n)$. With the assumption that $t_d = \bigO(|V_{\DBD}|^{1-\epsilon}f(|q|))$ and $t_u = \bigO(|V_{\DBD}|^{1-\epsilon}f(|q|))$ for some $\epsilon > 0$ and some function $f$, this leads to a total running-time of $\bigO(n^2 |V_{\DBD}|^{1-\epsilon}f(|q|)) = \bigO(n^{3-\epsilon})$.
\item Assume that there is an algorithm that solves $\enumProb$ with preprocessing $t_p$, delay $t_d$ and update time $t_u$. We can then solve $\TriProb$ on some instance $G = (V, E)$ with $V = \{v_1, v_2, \ldots, v_n\}$ as follows (the result will then follow from the combinatorial subcubic equivalence of $\BMMProb$ and $\TriProb$ (see Section~\ref{sec:fineGrainedComplexity})). We construct the graph database $\DBD_{v_1}$ over $\{\ta\}$ with nodes $V_{\DBD_{v_1}} = \{s, t, (u, i) \mid 0 \leq i \leq 3, u \in V\}$ and arcs $E_{\DBD_{v_1}} = \{((u, i), \ta, (u', i+1)) \mid 0 \leq i \leq 2, (u, u') \in E\} \cup \{(s, \ta, (v_1, 0)), ((v_1, 3), \ta, t)\}$. Furthermore, we define the $\RPQ$ $q = \ta\ta\ta\ta\ta$. It can be easily seen that $q(\DBD_{v_1}) = \{(s, t)\}$ if there is a triangle that contains $v_1$, and $q(\DBD_{v_1}) = \emptyset$ otherwise. In particular, this means that enumerating $q(\DBD_{v_1})$ can be done in time $\bigO(t_d)$. \par
Analogously, by constructing $\DBD_{v_2}, \ldots, \DBD_{v_n}$ and enumerating $q(\DBD_{v_2}), \ldots, q(\DBD_{v_n})$, we can therefore check whether $G$ contains a triangle. For every $i$ with $2 \leq i \leq n$, we can obtain $\DBD_{v_{i+1}}$ from $\DBD_{v_{i}}$ by performing a constant number of updates (we just have to remove arcs $(s, \ta, (v_i, 0))$ and $((v_i, 3), \ta, t)$ and instead add arcs $(s, \ta,  (v_{i+1}, 0))$ and $((v_{i+1}, 3), \ta, t)$). Consequently, this whole procedure requires the preprocessing for a graph database with $\bigO(|V|)$ nodes and $\bigO(|E|)$ arcs, $|V|$ many delays and $|V|$ many updates, i.\,e., the total running time of this procedure is $\bigO(t_p + |V|(t_d + t_u))$. With the assumption that $t_p = \bigO(|V_{\DBD}|^{3-\epsilon}f(|q|))$ and $t_d = t_u = \bigO(|V_{\DBD}|^{2-\epsilon}f(|q|))$ for some $\epsilon > 0$ and some function $f$, this leads to a total running-time of $\bigO(|V|^{3-\epsilon}f(|q|) + |V||V_{\DBD}|^{2-\epsilon}f(|q|)) = \bigO(|V|^{3-\epsilon})$. \qedhere
\end{enumerate}
\end{proof}

The first bound rules out that we can get below $\bigO(|V_{\DBD}|)$ for delay \emph{and} update time, regardless of the preprocessing; the second one analogously rules out anything below $\bigO(|\DBD|)$ (for combinatorial algorithms and linear preprocessing). While all these lower bounds suggest that improving on the linear delay may be rather difficult, they leave the following case open.

\begin{qu}\label{openProblem}
Can $\enumProb$ be solved with $\bigO(|\DBD|)$ preprocessing and $\bigO(|V_{\DBD}|)$ delay in data complexity?
\end{qu}

What is a reasonable conjecture with respect to this question? A combinatorial algorithm that answers it in the positive does not seem to have any unlikely consequences. Indeed, it would just entail an $\bigO(|\DBD||V_{\DBD}|)$ algorithm for $\evalProb$ (which exactly fits to Theorem~\ref{crossProductUpperBoundEvalTheorem}), an $\bigO(|V||G|)$ algorithm for computing transitive closures and an $O(n^3)$ algorithm for multiplying Boolean $n \times n$ matrices, which is the state of the art for combinatorial algorithms (this is due to the obvious reductions from these problems to $\enumProb$). However, since Question~\ref{openProblem} is about enumeration, a positive answer also means that after $\bigO(|G|)$ preprocessing, we can enumerate the transitive closure of a graph with delay $\bigO(|V|)$, and that after $\bigO(n^2)$ preprocessing, all $1$-entries of the Boolean matrix multiplication can be enumerated with delay $\bigO(n)$. Are such enumeration algorithms unlikely, so that we should rather expect a negative answer to Question~\ref{openProblem}? In fact not, since for the simple $\RPQ$s $q = \ta^*$ or $q = \ta \ta$, which are sufficient to encode transitive closures and Boolean matrix multiplications, linear preprocessing and delay $\bigO(|V_{\DBD}|)$ is indeed possible, as we shall obtain as byproducts of the results in the next section.\par
We close this section by the following remark that points out some similarity (and differences) of reductions used in Sections~\ref{sec:nonEnumBounds}~and~\ref{sec:enumBounds}.

\begin{rem}
As already mentioned in Section~\ref{sec:booleTestWitness}, the reduction used for the combined complexity lower bounds of Theorems~\ref{CCCheckLowerOVBoundsTheorem} (and therefore Point $1$ of Theorem~\ref{simpleEnumLowerBounds}) is similar to the reduction from~\cite{BackursIndyk2016} used for proving conditional lower bounds of regular expression matching. Also note that~\cite{EquiEtAl2021} improves the bound of~\cite{BackursIndyk2016} by showing that sub-quadratic algorithms for finding a path in a graph labelled by a given pattern string are most likely impossible, even if we are allowed to build indexes in polynomial time. The paper~\cite{EquiEtAl2019} strengthens the bound of~\cite{BackursIndyk2016} by restricting the structure of the graphs. Moreover, the $\OV$-lower bound for $\countProb$ of Theorem~\ref{OVLowerBoundCountTheorem} is similiar to a lower bound on counting the results of certain conjunctive queries from~\cite{BerkholzEtAl2017}, and the $\OMv$-lower bound from Point~$1$ of Theorem~\ref{enumOMvTriangleLowerBoundTheorem} is similiar to a lower bound on enumerating certain conjunctive queries from~\cite{BerkholzEtAl2017}. \par
The quite simple observation that Boolean matrix multiplication can be expressed as querying a bipartite graph (with conjunctive queries) has also been used in~\cite{BaganEtAl2007} (see also~\cite[Section~6]{BerkholzEtAl2020}) and is also the base for the $\OMv$-lower bound of~\cite{BerkholzEtAl2017}. In the context of this paper, this connection has been used in the bounds of Theorems~\ref{BMMLowerBoundEvalTheorem},~\ref{SBMMLowerBoundEvalTheorem} and Points~$3$ and~$4$~of~\ref{simpleEnumLowerBounds}. \par
The obvious connection between evaluating (non-acyclic) conjunctive queries and finding triangles (or larger cliques) has already been observed in~\cite{BraultBaron2013} (see also~\cite[Section~6]{BerkholzEtAl2020}). However, the $\TriProb$-lower bounds of this paper are quite different, since $\RPQ$s cannot explicitly express the structure of a triangle (or a larger clique, for that matter) by using conjunction. Therefore, our respective lower bounds (Theorems~\ref{CCCheckLowerTriangleBoundsTheorem}, Point $2$ of Theorem~\ref{simpleEnumLowerBounds}, and Point $2$ of Theorem~\ref{enumOMvTriangleLowerBoundTheorem}) need to encode this aspect in a different way. With respect to Theorems~\ref{CCCheckLowerTriangleBoundsTheorem} and Point $2$ of Theorem~\ref{simpleEnumLowerBounds} this is done by using non-constant queries (which explains why the lower bounds are not for data complexity, in contrast to the case of conjunctive queries), and with respect to Point $2$ of Theorem~\ref{enumOMvTriangleLowerBoundTheorem}, which \emph{does} work for data complexity, it is done by using updates. Moreover, our $\TriProb$-lower bounds do not seem to extend to larger cliques like it is the case for conjunctive queries (see \cite[Section~6]{BerkholzEtAl2020}).\par
Finally, we wish to point out that although some of the reductions used in this paper are similar to reductions used in the context of conjunctive queries, due to the difference of $\RPQ$s and $\CQ$s, none of the lower bounds directly carry over. Furthermore, note that the lower bound reductions in~\cite{BaganEtAl2007,BerkholzEtAl2017} have been used for obtaining dichotomies and therefore have been stated in a much more general way.
\end{rem}

\section{Enumeration with Sub-Linear Delay}

We now explore three different approaches towards enumeration of $q(\DBD)$ with delay strictly better than $\bigO(|\DBD|)$ (in data complexity): 
\begin{enumerate}
\item allowing super-linear preprocessing,
\item enumerating a representative subset of $q(\DBD)$, and 
\item restricting the $\RPQ$s.
\end{enumerate}

\subsection{Super-Linear Preprocessing}

Regarding the first approach, we can improve the delay from $\bigO(|\DBD|)$ to $\bigO(|V_{\DBD}|)$ by increasing the linear preprocessing by a factor of $\avgdegree(\DBD) \log(\avgdegree(\DBD))$ (in data complexity). Recall that $\avgdegree(\DBD)$ is the average degree of $\DBD$.

\begin{thm}\label{crossProductEnumSuperlinearPreprocTheorem}
Sorted $\enumProb$ can be solved with  $\bigO(|q|^2 \log(\avgdegree(\DBD)|q|)\avgdegree(\DBD)|\DBD|)$ preprocessing and $\bigO(|V_{\DBD}|)$ delay.
\end{thm}

\begin{proof}
Let $\DBD$ be a graph database over $\Sigma$, let $q$ be an $\RPQ$ over $\Sigma$ and let $\preceq$ be the order on $V_{\DBD}$. We assume that $V_{\DBD} = [n]$ with $1 \preceq 2 \preceq \ldots \preceq n$ (see Lemma~\ref{makeWellFormedLemma}). Recall that $\avgdegree(\DBD)$ denotes the average degree of $\DBD$, which in particular means that $\bigO(|E_{\DBD}|)=\bigO(|V_{\DBD}|\avgdegree(\DBD))$.
\medskip\\
\noindent\textbf{Preprocessing:} 
\begin{enumerate}
\item\label{SLStepOne} Compute $\DBtimesRE{\DBD}{q}$ and interpret it as its underlying non-labelled graph. We also compute the number $\avgdegree(\DBD)$, which can be done by moving through the adjacency lists of $\DBD$ and counting all edges.
\item\label{SLStepTwo} Let $H = (V_H, E_H)$ be the DAG of $\DBtimesRE{\DBD}{q}$'s strongly connected components. Without loss of generality, we assume that connected components are numbered with $1, 2, \ldots, \ell$ for some $\ell \in \mathbb{N}$, i.\,e., $V_H = [\ell]$, and we further assume that $1, 2, \ldots, \ell$ is a reverse topological sorting of $H$. We can compute $H$ as follows.
\begin{enumerate}
\item\label{SLStepTwoA} By the algorithm of Tarjan~\cite{Tarjan72}, we first compute the strongly connected components $K_1, K_2, \ldots, K_{\ell} \subseteq \DBtimesREV{\DBD}{q}$ of $\DBtimesRE{\DBD}{q}$. Since the finishing times of the depth-first searches of this algorithm determine a reverse topological sorting of $H$, we can assume that $K_1, K_2, \ldots, K_\ell$ describes a reverse topological sorting. We represent the strongly connected components $K_1, K_2, \ldots, K_{\ell}$ by an array $L$ of size $|V_{\DBD}||q|$ addressable by the nodes of $\DBtimesREV{\DBD}{q}$, i.\,e., for every $(i, p) \in \DBtimesREV{\DBD}{q}$, we set $L[(i, p)] = j$ if and only if $(i, p) \in K_j$.
\item\label{SLStepTwoB} In order to construct the actual graph $H$, we first initialise an empty adjacency list for every node $j \in V_H = [\ell]$ and also an additional empty list $P_j$. Then we iterate through the list that represents $\DBtimesREV{\DBD}{q}$ and add every node $(i, p)$ that we encounter to the list $P_{L[(i,p)]}$. For every $j \in [\ell]$, $P_{j}$ now represents $K_j$ as a list. Next, we initialise an empty integer array $W$ of size $\ell$. Then, for each $j \in [\ell]$, we iterate through the list $P_j$ and for each node $(i,p)$ that we encounter, we do the following. For every $(i', p')$ in the adjacency list of $(i, p)$, if $W[L[i', p']] \neq j$, then we add $L[(i', p')]$ to the adjacency list of $j$ and set $W[L[i', p']] = j$ (the array $W$ makes sure that we do not add the same arc between strongly connected components more than once, i.\,e., for each individual $j \in [\ell]$ considered in the above procedure, $W$ stores for each strongly connected component $j' \in [\ell]$ whether we have already added the arc $(j, j')$). 
\end{enumerate}
\item\label{SLStepThree} Construct $H^R$.
\item\label{SLStepFour} Initialise for each node $j$ in $H^R$ an empty AVL tree $A_j$ (or any search tree that allows the operations insert, delete, lookup, delete-min and delete-max in worst-case time that is logarithmic in its size), and store pointers to the roots of these trees and a counter for their size in an array $A$ of size $\ell$, such that finding the tree associated to a specific node can be done in constant time.
For every $i \in [n]$, insert $i$ into the tree $A_{L[(i,p_f)]}$, and at each point where the size of $A_{L[(i,p_f)]}$ exceeds $\avgdegree(\DBD)|q|$ during these insertions, delete the node of largest index from $A_{L[(i,p_f)]}$. 
\item\label{SLStepFive} For  $j = 1, \dots, \ell$ go through the adjacency list of node $j$ in $H^R$. For each arc $(j,j')$ in this list, insert one by one all entries from $A_j$ into $A_{j'}$ (first lookup to avoid duplicates). Update also the size counter of the tree  $A_{j'}$ accordingly, and at each point where this size exceeds $\avgdegree(\DBD)|q|$ during these insertions, delete the node of largest index from $A_{j'}$. This ensures that $A_{j'}$ is kept at a size at most $\avgdegree(\DBD)|q|$.
\item\label{SLStepSix}  Initialise a Boolean array $S$ of size $\ell$ with entries $0$ everywhere. Create for each $j\in[\ell]$  a sorted array $Z_j$  of all entries in $A_{j}$ by successively extracting and deleting the smallest entry. If $Z_j$  is not empty, set $S[j]=1$.
\item\label{SLStepSeven} Initialise a Boolean array $T$ of size $|V_{\DBD}|$ with entries $0$ everywhere.
\end{enumerate}\medskip
\noindent\textbf{Enumeration:}
In the enumeration phase, we carry out the following procedure. 
\begin{itemize}
\item For every $i = 1, 2, \ldots, n$:
\begin{itemize}
\item If $S[L[(i,q_0)]] = 1$, then 
\begin{itemize}
\item If $Z_{L[(i,q_0)]}$ has size less than $\avgdegree(\DBD)|q|$, output all $(i,j)$ for all $j$ stored in $Z_i$ (in the sorted order).
\item If  $Z_{L[(i,q_0)]}$ has size  $\avgdegree(\DBD)|q|$ perform a BFS in $\DBtimesRE{\DBD}{q}$ starting at node $(i, p_0)$ and for every visited $(i', p_f)$, set $T[i'] = 1$. Output during this BFS, all pairs $(i,i')$ for all $i'$ stored in $Z_{L[(i,q_0)]}$ (in the sorted order) to keep the desired delay, i.e., start with $k = 1$ and whenever the delay of $|V_{\DBD}||q|$ has expired output $(i,Z_i[k])$ and set $k=k+1$.
\item for every $i' = Z_{L[(i,q_0)]}[k], \ldots, n$: produce $(i, i')$ as output if $T[i'] = 1$,
\item Set all entries of $T$ to $0$.
\end{itemize}
\end{itemize}
\end{itemize}
\noindent\textbf{Running time:} We first note that Steps~\ref{SLStepOne}~to~\ref{SLStepThree} can be done in time $\bigO(|\DBtimesRE{\DBD}{q}|) = \bigO(|\DBD||q|)$: Lemma~\ref{crossProductSizeLemma} ensures that $\DBtimesRE{\DBD}{q}$ can be computed in $\bigO(|\DBD||q|)$, Lemma~\ref{reversalLemma} ensures that $H^R$ can be computed in $\bigO(|H|) = \bigO(|\DBD||q|)$, and it can be easily seen that in order to construct $H$, we only have to move through $\DBtimesRE{\DBD}{q}$ a constant number of times; thus, time $\bigO(|\DBtimesRE{\DBD}{q}|) = \bigO(|\DBD||q|)$ is sufficient.\par
Since the search trees never exceed a size of $\avgdegree(\DBD)|q|$, each operation supported by the search trees can be carried out in time $\bigO(\log(\avgdegree(\DBD)|q|))$. In Step~\ref{SLStepFour}, we first initialise the search trees, which can be done in $\bigO(|V_H|) = \bigO(|\DBtimesREV{\DBD}{q}|)$. Then, adding $i$ to the search tree $A_{L[(i,p_f)]}$ for every $i \in [n]$ can be done in time $\bigO(|V_{\DBD}|\log(\avgdegree(\DBD)|q|))$. In Step~\ref{SLStepFive}, we have to perform $\bigO(\avgdegree(\DBD)|q|)$ search tree operations per arc of $H$, i.\,e., we need time $\bigO(|E_H|\avgdegree(\DBD)|q| \log(\avgdegree(\DBD)|q|)) = \bigO(|\DBtimesRE{\DBD}{q}|\avgdegree(\DBD)|q| \log(\avgdegree(\DBD)|q|))$.\par
Step~\ref{SLStepSix}, we have to copy, for every $j \in [\ell]$, the elements from $A_j$ to $Z_j$, which requires a total of $\ell \avgdegree(\DBD)|q|$ search tree operations, so time $\bigO(|\DBtimesRE{\DBD}{q}|\avgdegree(\DBD)|q| \log(\avgdegree(\DBD)|q|))$ is sufficient (note that the time needed for initialising and filling $S$ is clearly subsumed by this). Finally, Step~\ref{SLStepSeven} can obviously be carried out in time $\bigO(|V_{\DBD}|)$. Hence, the total running time for the preprocessing is $\bigO(|\DBtimesRE{\DBD}{q}|\avgdegree(\DBD)|q| \log(\avgdegree(\DBD)|q|)) = \bigO(|\DBD||q|^2\avgdegree(\DBD)\log(\avgdegree(\DBD)|q|))$. \medskip\\
\noindent\textbf{Correctness:} We first show that after the preprocessing terminates, for every $i \in [n]$, $Z_{L[(i, p_0)]}$ contains the first $\avgdegree(\DBD)|q|$ nodes $i'$ (``first'' in the sense of smallest index) such that the node $(i', p_f)$ is reachable from  $(i, p_0)$. In other words, $Z_{L[(i, p_0)]}$ contains the first $\avgdegree(\DBD)|q|$ elements of a lexicographically sorted enumeration of $\{(i, i') \mid i' \in [n], (i, i') \in q(\DBD)\}$. To this end, we will show if $i'$ is among the $\avgdegree(\DBD)|q|$ smallest reachable nodes from $(i, p_0)$, then $i'$ is at some point inserted into $A_{L[i,p_0]}$, which means that $i'$ is among the nodes that are copied to $Z_{L[(i, p_0)]}$ in Step~\ref{SLStepSix} (recall that we only remove largest elements from some $A_j$ if its size exceeds $\avgdegree(\DBD)|q|$).\par
If $(i, p_0)$ and $(i', p_f)$ are in the same strongly connected component in $\DBtimesRE{\DBD}{q}$, then $L[(i, p_0)]=L[(i', p_f)]$ which means that in Step~\ref{SLStepFour} the index $i'$ is inserted into $A_{L[(i', p_f)]}=A_{L[(i, p_0)]}$. Otherwise there is a path in $H$ from $L[(i, p_0)]$ to $L[(i', p_f)]$, and reversing this path yields a path from $L[(i', p_f)]$ to  $L[(i, p_0)]$ in $H^R$. Denote by $L[(i', p_f)]=i_1, i_2,\dots, i_x=L[(i, p_0)]$ the vertices on the  path from $L[(i', p_f)]$ to  $L[(i, p_0)]$ in $H^R$. By the reverse topological sorting on $H$, it follows that $i_1< i_2< \ldots < i_x$. In Step~\ref{SLStepFour}, $i'$ is put in $A_{i_1}=A_{L[(i', p_f)]}$, and iteratively in Step~\ref{SLStepFive}, $i'$ is copied from $A_{i_h}$ to $A_{i_{h+1}}$ for $1\leq h < x$, and therefore also into $A_{i_x} = A_{L[i,p_0]}$. \par
In the enumeration procedure, we go through \emph{phases} $i  = 1, 2, \ldots, n$ and we will next show that in phase $i$, we enumerate without duplicates exactly the pairs $(i, i') \in q(\DBD)$ sorted by their right element (this means that the enumeration procedure enumerates $q(\DBD)$ without duplicates and lexicographically sorted). \par
If the size of $Z_{L[(i, p_0)]}$ is less than $\avgdegree(\DBD)|q|$, then, by the considerations above, we know that $Z_{L[(i, p_0)]}$ contains \emph{all} $i'$ with $(i, i') \in q(\DBD)$. Thus, we enumerate all pairs $(i, i')$ with $i' \in Z_{L[(i, p_0)]}$ sorted by their right elements. On the other hand, if the size of $Z_{L[(i, p_0)]}$ is at least $\avgdegree(\DBD)|q|$, we perform a BFS starting in $(i, p_0)$. After termination of this BFS, we output exactly the pairs $(i, i')$ for which $(i', p_f)$ is visited in this BFS and that are not already contained in $Z_{L[(i, p_0)]}$. This shows that we correctly enumerate $q(\DBD)$ without duplicates. Furthermore, since we list again the pairs $(i, i')$ by increasing index~$i'$ (the first $\avgdegree(\DBD)|q|$ by the sorting provided by $Z_{L[(i, p_0)]}$, and the possibly further by the sorted array $T$), the enumeration is again sorted by the right elements of the pairs. \par
It only remains to estimate the delay of the enumeration procedure, and it is sufficient to do this for each phase $i$ separately. If $S[L[(i,q_0)] = 0$ then phase $i$ terminates in constant time, which also means that in case that $S[L[(i,q_0)] = 0$ for all phases $i$, we need at most time in $O(|V_{\DBD}|)$ to determine that the enumeration has terminated. If the size of $Z_{L[(i, p_0)]}$ is less than $\avgdegree(\DBD)|q|$, then we enumerate all pairs with constant delay. If, on the other hand, the size of $Z_{L[(i, p_0)]}$ is $\avgdegree(\DBD)|q|$, then this array contains enough solutions to pay for the BFS we run for $(i,p_0)$, in the sense that producing the output form $Z_{L[(i, p_0)]}$ with delay $|V_{\DBD}|$ yields time $|V_{\DBD}|\avgdegree(\DBD)|q|=|E_{\DBD}||q|$ which is enough to completely run the BFS. The possibly further solutions stored in $T$ are then produced with delay at most $|V_{\DBD}|$ (for running through this list to find the non-zero entries).
\end{proof}

Obviously, in the worst case we can have $\avgdegree(\DBD) = \Omega(|V_{\DBD}|)$ and then the preprocessing of the algorithm of Theorem~\ref{crossProductEnumSuperlinearPreprocTheorem} is $\Omega(|V_{\DBD}||\DBD|)$ (in data complexity) and therefore no improvement over just computing the complete set $q(\DBD)$ in time $\bigO(|V_{\DBD}||\DBD|)$. However, for graph databases with low average degree, we can decrease the delay significantly from $\bigO(|\DBD|)$ to $\bigO(|V_{\DBD}|)$ at the cost of a slightly super-linear preprocessing time. As another remark about Theorem~\ref{crossProductEnumSuperlinearPreprocTheorem}, we observe that the pre-computed information becomes worthless if $\DBD$ is updated. \par
Finally, we discuss a minor modification of the algorithm of Theorem~\ref{crossProductEnumSuperlinearPreprocTheorem}, which yields a slightly different result that is interesting in its own right (we state this result as a corollary after the remark). 

\begin{rem}\label{lazyStuff}
At the cost of additional space (in $O(|V_{\DBD}|^2)$) and giving up on sorted enumeration, we can use the lazy array initialisation technique (see, e.\,g., the textbook~\cite{MoretShapiroBook}) in order to reduce the preprocessing time in Theorem~\ref{crossProductEnumSuperlinearPreprocTheorem} by a factor $\log(\avgdegree(\DBD)|q|)$ to $\bigO(|q|^2 \avgdegree(\DBD)|\DBD|)$.
This can be done by replacing each AVL-tree $A_j$ in step~\ref{SLStepFour} by an array $A_j$ (with lazy initialization) of size $|V_{\DBD}|$, and completely dropping step~\ref{SLStepSix}. Instead of the arrays $Z_j$, we use in the enumeration phase the unsorted lists $\widehat{A}_j$ of the elements stored in the arrays~$A_j$. In step~\ref{SLStepFour} and~\ref{SLStepFive}, we simply stop filling the arrays $A_j$ as soon as their counter exceeds $\avgdegree(\DBD)|q|$.
This reduces the running time for these steps to $\bigO(|V_{\DBD}|)$, and $\bigO(|\DBtimesRE{\DBD}{q}|\avgdegree(\DBD)|q|)$, respectively. Note that the lazy initialization allows us to avoid duplicates in constant time (which before was achieved by looking up the elements in the AVL-trees which required time $\bigO(\log(\avgdegree(\DBD)|q|))$). In every phase $i$ of the enumeration procedure instead of a sorted list $Z_{L[(i, p_0)]}$ with the (at most) $\avgdegree(\DBD)|q|$ smallest nodes $i'$ with $(i, i') \in q(\DBD)$ we only have an unsorted list $\widehat{A}_{L[(i, p_0)]}$ of at most $\avgdegree(\DBD)|q|$ nodes $i'$ (not necessarily the smallest ones) with $(i, i') \in q(\DBD)$. In the case where less than $\avgdegree(\DBD)|q|$ elements are stored in $\widehat{A}_{L[(i, p_0)]}$, we can just produce all $(i, i')$ for all $i'$ stored in $\widehat{A}_{L[(i, p_0)]}$. On the other hand, if $\widehat{A}_{L[(i, p_0)]}$ stores more elements, then again we use the elements stored in $\widehat{A}_{L[(i, p_0)]}$ in order to pay for a BFS and store in the array~$T$ the new elements found by the BFS. However, since we do not produce the elements in a sorted way, we use $T$ to also track which elements have been produced already. Whenever we output $(i, i')$ with $i'$ from $\widehat{A}_{L[(i, p_0)]}$ during the run of the BFS for $i$, we set $T[i']=2$ (and never overwrite a 2-entry in $T$). In order to list the pairs produced by the BFS after its termination, we pass once through~$T$ and only produce  $(i, i')$ if $T[i']=1$.  As last technical detail, we change step~\ref{SLStepSeven} to initialise $T$ as integer array (so we can actually have 2-entries in $T$).
\end{rem}

\begin{cor}\label{crossProductEnumSuperlinearPreprocCorollary}
$\enumProb$ can be solved with $\bigO(|q|^2 \avgdegree(\DBD)|\DBD|)$ preprocessing and $\bigO(|V_{\DBD}|)$ delay.
\end{cor}

\subsection{Enumeration of Representative Subsets}

Evaluating an $\RPQ$ $q$ on a graph database $\DBD$ aims to find for each node $u$ \emph{all} its $q$-successors, i.\,e., nodes reachable by a $q$-path. It is therefore a natural restriction to ask for only \emph{at least one} (if any) such successor. Likewise, we could also ask for at least one (if any) $q$-predecessor for every node. More precisely, instead of the whole set $q(\DBD)$, the task is to enumerate a \emph{$q(\DBD)$-approximation}, which is a set $A \subseteq q(\DBD)$ such that, for every $u, v \in V_{\DBD}$, if $(u, v) \in q(\DBD)$, then also $(u, v'), (u', v) \in A$ for some $u', v' \in V_{\DBD}$. Such a set is representative for $q(\DBD)$, since it contains for \emph{every} node the information, whether it is involved as a source and whether it is involved as a target in some reachable pair from $q(\DBD)$. The problem of enumerating any $q(\DBD)$-approximation for given $\DBD$ and $q$ will be denoted by $\approxEnumProb$. 

\begin{lem}\label{computeAppLemma}
Given a graph database $\DBD$ and an $\RPQ$ $q$, we can compute a $q(\DBD)$-approximation in time $\bigO(|\DBD||q|)$.
\end{lem}

\begin{proof}
Let $\DBD$ be a graph database over $\Sigma$, let $q$ be an $\RPQ$ over $\Sigma$ and let $\preceq$ be the order on $V_{\DBD}$. We assume that $V_{\DBD} = [n]$ with $1 \preceq 2 \preceq \ldots \preceq n$ (see Lemma~\ref{makeWellFormedLemma}). \par
We compute $\DBtimesRE{\DBD}{q}$, which, according to Lemma~\ref{crossProductSizeLemma}, can be done in time $\bigO(|\DBD||q|)$, and we interpret $\DBtimesRE{\DBD}{q}$ as its underlying non-labelled graph. We construct two arrays $S$ and $T$ of size $|V_{\DBD}|$ such that, for every $i \in [n]$, $S[i] = (i, p_0)$ and $T[i] = (i, p_f)$. Note that this also means that pointers to the corresponding adjacency lists are stored along with the nodes in $S$ and $T$. Computing $S$ and $T$ can be done in time $\bigO(|V_{\DBD}||q|)$ as follows. We move through the list that represents $\DBtimesREV{\DBD}{q}$ and for every node $(i, p_0)$, we set $S[i] = (i, p_0)$, and for every node $(i, p_f)$, we set $T[i] = (i, p_f)$. \par
We modify $\DBtimesRE{\DBD}{q}$ by adding a new node $v_{\mathsf{source}}$ with an arc to each $(i, p_0)$ with $i \in [n]$, and a new node $v_{\mathsf{sink}}$ with an arc from each $(i, p_f)$ with $i \in [n]$. This can be done in time $\bigO(|V_{\DBD}|)$ as follows. We first add $v_{\mathsf{source}}$ and $v_{\mathsf{sink}}$ to the list that represents $\DBtimesREV{\DBD}{q}$, which requires constant time. Then we add all vertices from $S$ to the adjacency list for $v_{\mathsf{source}}$, which requires time $\bigO(|V_{\DBD}|)$. Finally, for every node $(i, p_f)$ in $T$, we add $v_{\mathsf{sink}}$ to the adjacency list for $(i, p_f)$. Again, this can be done in time $\bigO(|V_{\DBD}|)$. \par
Let $S'$ and $T'$ be arrays of size $|\DBtimesREV{\DBD}{q}|$ the entries of which can store values from $[n] \cup \{0\}$ and can be addressed by the nodes from $\DBtimesREV{\DBD}{q}$. We wish to fill these arrays such that they satisfy the following properties. For every $(u, p) \in \DBtimesREV{\DBD}{q}$, if $(u, p)$ is reachable from some node from $S$, then $S'[(u, p)] = i$ for some $i \in [n]$ such that $(u, p)$ is reachable from $(i, p_0)$, and if $(u, p)$ is not reachable from any node from $S$, then $S'[(u, p)] = 0$. Analogously, for every $(u, p) \in \DBtimesREV{\DBD}{q}$, if $(u, p)$ can reach some node from $T$, then $T'[(u, p)] = i$ for some $i \in [n]$ such that $(u, p)$ can reach $(i, p_f)$, and if $(u, p)$ cannot reach any node from $T$, then $T'[(u, p)] = 0$. \par
This means that $S'$ and $T'$ contain all the information we need to construct a $q(\DBD)$-approximation. More precisely, let 
\begin{equation*}
A = \{(i, T'[(i, p_0)]) \mid i \in [n], T'[(i, p_0)] \neq 0\} \cup \{(S'[(i, p_f)], i) \mid i \in [n], S'[(i, p_f)] \neq 0\}\,. 
\end{equation*}
It can be easily seen that $A \subseteq q(\DBD)$ and, for every $u, v \in V_{\DBD}$,
\[
(u, v) \in q(\DBD) \implies (\exists v' \in V_{\DBD}: (u, v') \in A) \wedge (\exists u' \in V_{\DBD}: (u', v) \in A)\,. 
\]
Thus, $A$ is a $q(\DBD)$-approximation. Moreover, provided that we have the arrays $S'$ and $T'$ at our disposal, we can compute $A$ in time $\bigO(|V_{\DBD}|)$. Consequently, in order to conclude the proof, it remains to explain how the arrays $S'$ and $T'$ can be computed in time $\bigO(|\DBD||q|)$.\par 
We assume that $S'$ and $T'$ are initialised with every entry storing $0$. For every $i \in [n]$, we set $S'[(i, p_0)] = i$ and $T'[(i, p_f)] = i$. This can be done in time $\bigO(|\DBtimesREV{\DBD}{q}|) = \bigO(|V_{\DBD}||q|)$. Next, we perform a BFS from $v_{\textsf{source}}$ and whenever we traverse an arc $((u, p), (u', p'))$, we set $S'[(u', p')] = S'[(u, p)]$. This can be done in time $\bigO(|\DBtimesRE{\DBD}{q}|) = \bigO(|\DBD||q|)$. Interpreting the initialisation that ensures $S'[(i, p_0)] = i$, for every $i \in [n]$, as the base of an induction and the assignments $S'[(u', p')] = S'[(u, p)]$ for traversed arcs $((u, p), (u', p'))$ as the step of the induction, we can directly conclude by induction that after termination of this BFS, $S'$ has the desired property. In order to compute $T'$, we first have to construct $(\DBtimesRE{\DBD}{q})^R$, which, according to Lemma~\ref{reversalLemma}, can be done in time $\bigO(|\DBtimesRE{\DBD}{q}|) = \bigO(|\DBD||q|)$. Then we perform a BFS from $v_{\textsf{sink}}$ and whenever we traverse an arc $((u, p), (u', p'))$, we set $T'[(u', p')] = T'[(u, p)]$. Hence, computing $S'$ and $T'$ can be done in time $\bigO(|\DBtimesRE{\DBD}{q}|) = \bigO(|\DBD||q|)$. 
\end{proof}

Lemma~\ref{computeAppLemma} directly implies the following result.

\begin{thm}\label{apprEnumTheorem}
$\approxEnumProb$ can be solved with $\bigO(|\DBD||q|)$ preprocessing and delay $\bigO(1)$.  
\end{thm}

\begin{proof}
Let $\DBD$ be a graph database over $\Sigma$ and let $q$ be an $\RPQ$ over $\Sigma$. In the preprocessing, we can compute a $q(\DBD)$-approximation $A$, which, according to Lemma~\ref{computeAppLemma}, can be done in time $\bigO(|\DBD||q|)$. In the enumeration, we simply enumerate the set $A$ with constant delay.
\end{proof}

Interestingly, it seems rather difficult to also support updates while keeping a low delay (the following bounds are due to the same reductions used for Theorem~\ref{enumOMvTriangleLowerBoundTheorem}, simply because these reductions produce instances for which all $q(\DBD)$-approximations are equal to $q(\DBD)$).

\begin{thm}\label{simpleApprLowerBounds}
If $\approxEnumProb$ can be solved with 
\begin{enumerate}
\item arbitrary preprocessing, $\bigO(|V_{\DBD}|^{1-\epsilon}f(|q|))$ updates and $\bigO(|V_{\DBD}|^{1-\epsilon}f(|q|))$ delay, then the $\OMv$-hypothesis fails.
\item $\bigO(|V_{\DBD}|^{3-\epsilon}f(|q|))$ preprocessing, $\bigO(|V_{\DBD}|^{2-\epsilon}f(|q|))$ updates and $\bigO(|V_{\DBD}|^{2-\epsilon}f(|q|))$ delay for some function $f$ and $\epsilon > 0$, then the $\combBMMProb$-hypothesis fails. 
\end{enumerate}
\end{thm}

\begin{proof}
We prove the two lower bounds of Theorem~\ref{simpleApprLowerBounds} separately. 
\begin{enumerate}
\item We observe that exactly the same reduction used in the proof of Point $1$ of Theorem~\ref{enumOMvTriangleLowerBoundTheorem} works as well. Since we assume the same bound on the time required for updates, we can construct the graph database $\DBD_{M, \vec{v}^1}$ in the same way as in the proof of Point $1$ of Theorem~\ref{enumOMvTriangleLowerBoundTheorem}. Recall that $q(\DBD_{M, \vec{v}^1}) = \{(u_i, w) \mid (M \vec{v}^1)[i]=1\}$. However, since for every $i \in [n]$, either no node is reachable by a path labelled with $11$ or only the node $w$ is reachable by a path labelled with $11$, we know that $A \subseteq q(\DBD)$ is a $q(\DBD_{M, \vec{v}^1})$-approximation if and only if $A = q(\DBD_{M, \vec{v}^1})$. Thus, we can compute $M \vec{v}^1$ by enumerating a $q(\DBD_{M, \vec{v}^1})$-approximation in time $\bigO(n t_d)$. Repeating this step in the same way as done in the proof of Point $1$ of Theorem~\ref{enumOMvTriangleLowerBoundTheorem} leads to an algorithm solving $\OMv$ in time $\bigO(n^{3-\epsilon})$.
\item We observe that exactly the same reduction used in the proof of Point $2$ of Theorem~\ref{enumOMvTriangleLowerBoundTheorem} works as well. We construct the graph database $\DBD_{v_1}$ over $\{\ta\}$ and the $\RPQ$ $q = \ta\ta\ta\ta\ta$ in the same way. Recall that $q(\DBD_{v_1}) = \{(s, t)\}$ if there is a triangle that contains $v_1$, and $q(\DBD_{v_1}) = \emptyset$ otherwise. This also means that any $q(\DBD_{v_1})$-approximation $A$ is equal to $\{(s, t)\}$ if there is a triangle that contains $v_1$, and $A = \emptyset$ otherwise. Thus, the assumption that $\approxEnumProb$ can be solved with preprocessing $\bigO(|V_{\DBD}|^{3-\epsilon}f(|q|))$, $\bigO(|V_{\DBD}|^{2-\epsilon}f(|q|))$ updates and $\bigO(|V_{\DBD}|^{2-\epsilon}f(|q|))$ delay for some function $f$ and $\epsilon > 0$, leads to combinatorial subcubic algorithm for $\TriProb$ in the same way as in the proof of Point $2$ of Theorem~\ref{enumOMvTriangleLowerBoundTheorem}.\qedhere
\end{enumerate}
\end{proof}

\subsection{Restricted \texorpdfstring{$\RPQ$s}{\textsf{RPQ}s}}

As our third approach, we show that for restricted $\RPQ$, we can solve $\enumProb$ with delay much smaller than $\bigO(|\DBD|)$. We first need some definitions. For any class $Q \subseteq \RPQ$, we denote by $\enumProbClass{Q}$ the problem $\enumProb$ where the input $\RPQ$ is from~$Q$. Moreover, $\bigvee(Q)$ is the set of all $\RPQ$s of the form $(q_1 \altop \ldots \altop q_m)$ with $q_i \in Q$ for every $i \in [m]$. An $\RPQ$ $q$ over $\Sigma$ is a \emph{basic transitive $\RPQ$} (over $\Sigma$) if $q = (x_1 \altop \ldots \altop x_k)^*$ or $q = (x_1 \altop \ldots \altop x_k)^+$, where $x_1, \ldots, x_k \in \Sigma$; and $q$ is a \emph{short $\RPQ$} (over $\Sigma$) if $q = (x_1 \altop \ldots \altop x_k)$ or $q = (x_1 \altop \ldots \altop x_k) (y_1 \altop \ldots \altop y_{k'})$, where $x_1, \ldots, x_k, y_1, \ldots, y_{k'} \in \Sigma$. By $\BTRPQ$ and $\SRPQ$, we denote the class of basic transitive $\RPQ$ and the class of short $\RPQ$, respectively. \par
We show that for the class $\bigvee(\SRPQ \cup \BTRPQ)$ (which, e.\,g., contains $\RPQ$s of the form $q = (\ta \tb \altop \tc^* \altop \tb (\tc \altop \td) \altop (\ta \altop \tb \altop \td)^+)$), semi-sorted $\enumProb$ can be solved with linear preprocessing and delay $\bigO(\Delta(\DBD))$ in data complexity (recall \emph{semi-sorted} from Section~\ref{sec:mainDefs}). To this end, we first prove a general upper bound for semi-sorted $\enumProbClass{\bigvee(Q)}$ in terms of the $\enumProbClass{Q}$. 

\begin{lem}\label{alternationConstructionLemma}
Let $Q \subseteq \RPQ$ be such that semi-sorted $\enumProbClass{Q}$ can be solved with preprocessing $p(|\DBD|, |q|)$ and delay $d(|\DBD|, |q|)$, then semi-sorted $\enumProbClass{\bigvee(Q)}$ can be solved with preprocessing $\bigO((|q| p(|\DBD|, |q|)) + |V_{\DBD}|)$ and delay $\bigO(|q| d(|\DBD|, |q|))$.
\end{lem}

\begin{proof}
We assume that semi-sorted $\enumProbClass{Q}$ can be solved with preprocessing $p(|\DBD|, |q|)$ and delay $d(|\DBD|, |q|)$. Let $q \in \bigvee(Q)$, i.\,e., $q = (q_1 \altop q_2 \altop \ldots \altop q_m)$ with $q_i \in Q$ for every $i \in [m]$. By assumption, there are enumeration algorithms $A_1, A_2, \ldots, A_m$ that, for any input graph database $\DBD$, perform some preprocessing in time $p(|\DBD|, |q_i|)$ and then enumerate the pairs $q_i(\DBD)$ ordered by their first elements with delay $d(|\DBD|, |q_i|)$. We describe now an enumeration algorithm $A$ that, for any input graph database $\DBD$, performs some preprocessing in time $\bigO(|q|p(|\DBD|, |q|))$ and then enumerates the pairs in $q(\DBD)$ ordered by their first elements with delay $\bigO(|q|d(|\DBD|, |q|))$. Note that $q(\DBD) = \bigcup^m_{i = 1} q_i(\DBD)$.\par
Let $\DBD$ be the input graph database over $\Sigma$. We assume that $V_{\DBD} = [n]$ with $1 \preceq 2 \preceq \ldots \preceq n$ (see Lemma~\ref{makeWellFormedLemma}).
 \medskip\\
\noindent\textbf{Intuitive Explanation}: The algorithm $A$ uses size-$m$ arrays $L$ and $R$ that store nodes from $V_{\DBD} = [n]$, and size-$n$ Boolean arrays $X$ and $Y$ that represent subsets of $[n]$. We construct these in the preprocessing, in addition to performing the preprocessing procedures of the algorithms $A_1, A_2, \ldots, A_m$. 
Since the enumeration algorithms $A_1, A_2, \ldots, A_m$ provide a semi-ordered enumeration, each enumeration of $q_j(\DBD)$ can be seen as a table with two columns whose rows correspond to the pairs from $q_j(\DBD)$, ordered by the first column. The enumeration then produces the rows of this table from top to bottom. Therefore, the enumeration proceeds in $n$ phases, where the $i^{\text{th}}$ phase consists in producing all rows with $i$ in the left column (note that for some $i$ such rows might not exist). Consequently, we can also interpret phase $i$ of the algorithms $A_j$ as an enumeration of only single elements from $V_{\DBD}$, since the left elements of the pairs are always $i$. The enumeration procedure of $A$ also proceeds in $n$ phases, where the $i^{\text{th}}$ phase consists in having the enumeration algorithms $A_1, A_2, \ldots, A_m$ perform their $i^{\text{th}}$ phases in an interleaved manner, i.\,e., we let each $A_i$ produce the next element one after the other, and then we repeat this step until all $A_j$ have finished their $i^{\text{th}}$ phase. We maintain in array $L$ the current phases (i.\,e., the current left element of the produced pairs) of the algorithms $A_j$ (this is necessary, since only a subset may participate in phase $i$ for a fixed $i$ and some $A_j$ may finish it earlier than others) and we store in array $R$ the elements produced most recently by the algorithms $A_j$. Once all $A_j$ (that are still in phase $i$) have produced one new element, we disregard all those elements $v$ among them such that $A$ has already produced $(i, v)$ in phase $i$ (these already produced elements will be stored in $Y$), and we store all other elements in $X$. Then $A$ picks some element $v \in X$, produces $(i, v)$ (and therefore ``buys'' another delay), marks $v$ as already produced by adding it to $Y$ and removes it from $X$. Then, again each $A_j$ will produce the next element and this goes on until all $A_j$ have finished phase $i$, and therefore we can move on to the next phase. \par
Due to the sets $X$ and $Y$, we do not produce duplicates. However, in order to bound the delay, we have to produce at least one pair in each iteration of the main loop in each phase, and therefore we have to show that it cannot happen that $X$ is empty when $A$ needs to produce the next pair. We shall now define this algorithm more formally. \medskip\\
\noindent\textbf{Preprocessing}: First, $A$ performs all preprocessing procedures of the algorithms \[A_1, A_2, \ldots, A_m\,,\] which can be done in time $\sum^{m}_{i = 1} p(|\DBD|, |q_i|) = \bigO(|q| p(|\DBD|, |q|))$. Then we construct arrays $L$ and $R$ each of size $m$ the entries of which can store values from $[n] \cup \{0, n+1\}$, and we initialise all entries with $0$. 
Furthermore, we construct Boolean arrays $X$ and $Y$ of size $n$ that initially store $0$ in every entry (note that these arrays shall be used for storing subsets of $[n]$). This can clearly be done in time $\bigO(|V_{\DBD}| + |q|)$. Consequently, the total preprocessing time is $\bigO(|q| p(|\DBD|, |q|) + |V_{\DBD}|)$.\medskip\\
\noindent\textbf{Enumeration}: 
Firstly, we set $c = 0$, where $c$ will be a counter that indicates the current phase. Then we iterate the following main loop until it terminates:
\begin{enumerate}
\item\label{StepOne} For every $i \in [m]$ with $L[i] = c$, we request the next element from the enumeration procedure of $A_i$. If such an element $(u, v)$ is returned, we set $L[i] = u$ and $R[i] = v$, and if no element is returned (i.\,e., the enumeration of $A_i$ is done), we set $L[i] = R[i] = n+1$. 
\item\label{StepTwo} If $c < \min\{L[i] \mid i \in [m]\}$, then 
\begin{enumerate}
\item\label{StepTwoA} for every $v \in X$, remove $v$ from $X$ and produce $(c, v)$,
\item\label{StepTwoB} set $Y = \emptyset$,
\item\label{StepTwoC} set $c = \min\{L[i] \mid i \in [m]\}$.
\end{enumerate}
\item\label{StepInterrupt} If $c = n + 1$, interrupt.
\item\label{StepThree} For every $i \in [m]$, if $L[i] = c$ and $R[i] \notin Y$, add $R[i]$ to $X$.
\item\label{StepFour} Choose some element $v \in X$, produce $(c, v)$ as output, remove $v$ from $X$ and add $v$ to $Y$.
\item\label{StepFive} Move to Step~\ref{StepOne}.
\end{enumerate}
\medskip
\noindent\textbf{Correctness}: We first note that in the first iteration of the main loop, we have $L[i] = c = 0$ for every $i \in [m]$, so we request the first pairs from \emph{all} $A_j$ and store the respective left and right elements in $L$ and $R$, respectively (or store $n+1$ in both $L$ and $R$, if the enumeration procedure has already terminated without producing any pair). In any subsequent iteration, this only happens with respect to those $A_j$ that are in phase $c$. If condition $c < \min\{L[i] \mid i \in [m]\}$ is satisfied in Step~\ref{StepTwo}, then this means that by performing Step~\ref{StepOne}, we have reached a new phase (note that also in the very first iteration, we reach a new phase, i.\,e., phase $i$ for the smallest left element $i \in [n]$ in any pair produced by some $A_j$) and therefore we produce all elements of $X$ as output (paired with $c$) and then empty the sets $X$ and $Y$ (which were only storing elements that are relevant for phase $c$ that just terminated). Moreover, we have to determine the next phase, which is done by setting $c = \min\{L[i] \mid i \in [m]\}$. In the very first iteration, $X$ is empty, so nothing happens in these steps except setting $c$ to the current phase (which must be strictly larger than $0$, but not necessarily $1$).
Next, Step~\ref{StepInterrupt} would now interrupt the whole procedure, if $c = n + 1$, which means that all enumeration procedures of the algorithms $A_j$ have already terminated. Note that if this happens for the first time, i.\,e., the last still active enumeration procedures of some $A_j$ terminate by requesting their last elements in Step~\ref{StepOne}, then we necessarily also have $c < \min\{L[i] \mid i \in [m]\}$ in Step~\ref{StepTwo}, which means that we have produced all elements from $X$ as output before the interruption is invoked in Step~\ref{StepInterrupt}. If, on the other hand, we reach Step~\ref{StepThree}, then we collect in the set $X$ the right elements from all new pairs produced by $A_j$ that are in phase $c$, but only if these are not already stored in $Y$, since then they would have already been produced as right element in a pair with $c$ as left element. Step~\ref{StepFour} then chooses and actually produces (paired with $c$) one of the elements from $X$ (and removes it from $X$ and adds it to $Y$ to store that is has already been produced). Finally, in Step~\ref{StepFive}, we move back to Step~\ref{StepOne}, which triggers a new iteration.\par
These considerations show that the procedure from above will enumerate $q(\DBD)$. In particular, note that if we enter Step~\ref{StepFour} with $X$ being empty, then we do not produce an output at this point, but, due to Step~\ref{StepTwo}, we will nevertheless completely enumerate $q(\DBD)$. Since we synchronise the enumerations of the $A_j$ with respect to their phases, the enumeration of $q(\DBD)$ produces \emph{all} pairs from $q(\DBD)$ with left element $i$, then all pairs with left element $i' > i$ and so on; thus, the enumeration of $q(\DBD)$ is semi-ordered. Furthermore, the book keeping done in sets $X$ and $Y$ guarantees that we do not produce duplicates. It only remains to analyse the delay of this enumeration procedure.\par
We first observe that Step~\ref{StepTwoB}, i.\,e., setting $Y = \emptyset$, is problematic since it requires time $\bigO(n)$ (we have to set $Y[i] = 0$ for every $i \in [n]$). Therefore, we implement the array $Y$ as follows. Instead of letting it be Boolean, we assume that it can store elements from $[n] \cup \{0\}$. The idea is that $Y[i] = 0$ means that $i \notin Y$ (just as for the Boolean case), while $Y[i] = j$ with $j \in [n]$ means $i \in Y$ in the case that we are currently in phase $j$, i.\,e., $c = j$, and $i \notin Y$ otherwise. With this interpretation, Step~\ref{StepTwoB} is not necessary anymore, since setting $c = \min\{L[i] \mid i \in [m]\}$ in Step~\ref{StepTwoC} has the same effect as erasing all elements from $Y$. Consequently, we can ignore Step~\ref{StepTwoB} altogether (or rather interpret as a mere comment in the pseudo code above to indicate what is happening at Step~\ref{StepTwoB}). In particular, we note that with this implementation of $Y$, we can still check for both $X$ and $Y$ whether they contain a specific element, and we can add or erase specific elements in constant time (adding $i$ to $Y$ in phase $c$ just means to set $Y[i] = c$ instead of $Y[i] = 1$). \par
However, in Steps~\ref{StepTwoA}~and~\ref{StepFour}, we also have to retrieve \emph{some} element from $X$. In order to do this efficiently (and not by moving through the array from left to right to find some elements, which requires time $\bigO(n)$), we also store the elements of $X$ as an unsorted list (which is initialised in the preprocessing). This means that we can always obtain some element of $X$ in constant time (by just retrieving the first list element). Keeping the array and the list for $X$ synchronised is no problem: Whenever we add some $i$ to $X$ (Step~\ref{StepThree}), we set $X[i] = 1$ and add $x$ at the end of the list for $X$; whenever we want to retrieve some element from $X$ (Steps~\ref{StepTwoA}~and~\ref{StepFour}), we retrieve and remove the first list element, say $i$, and then we set $X[i] = 0$. Consequently, all operations with respect to the lists $X$ and $Y$ can be performed in constant time. \par
We estimate the running-time for each of the separate steps of an iteration. Step~\ref{StepOne} requires time $\bigO(\sum^m_{j = 1} d(|\DBD|, |q_j|)) = \bigO(m d(|\DBD|, |q|))$. The total running time of Steps~\ref{StepTwo}~to~\ref{StepTwoC} is $\bigO(m + k)$, where $k$ is the number of pairs $(c, v)$ produced in Step~\ref{StepTwoA} and $\bigO(m)$ is needed to compute $\min\{L[i] \mid i \in [m]\}$. Note that the pairs that are produced in Step~\ref{StepTwo} are output with constant delay in Step~\ref{StepTwoA} and pay for the running time dependence on $k$, hence in the worst case, $k=0$ where Step~\ref{StepTwo} requires time in $\bigO(m)$.
 Step~\ref{StepThree} requires time $\bigO(m)$. All other steps can be carried out in constant time. This means that if in each iteration at least one pair is produced by Step~\ref{StepFive}, then the delay of the whole enumeration procedure of algorithm $A$ is $\bigO(m d(|\DBD|, |q|))$. 
Obviously, if we can never reach the situation that $X = \emptyset$ in Step~\ref{StepFour}, then in each iteration at least one pair is produced. Therefore, it is sufficient to prove this property.\par
Let us assume that the enumeration procedure has just finished Step~\ref{StepInterrupt} and we are in some iteration of phase $i$, i.\,e., $c = i$. Moreover, we assume that so far, we have not encountered the situation that $X = \emptyset$ in Step~\ref{StepFour}. For every $j \in [m]$, let $a_j = R[j]$ if $L[j] = c$ where $a_j = \bot$ indicates the situation that element $a_j$ does not exist. This means that all existing elements $a_j$ with $j \in [m]$ are exactly those elements that have most recently been produced in phases $i$ of the enumeration procedures from the algorithms $A_j$ (this can have happened in Step~\ref{StepOne} of the same iteration or, if this is the first iteration of phase $i$, also in applications of Step~\ref{StepOne} in previous iterations). In particular, the existing elements $a_j$ with $j \in [m]$ have not yet been handeled in the sense of Step~\ref{StepThree}, i.\,e., we have not yet checked whether they have already been produced as output and, if not, have added them to $X$. We can also note that there must be at least one $j \in [m]$ with $a_j \neq \bot$, since otherwise $c < \min\{L[i] \mid i \in [m]\}$.\par
In addition to these elements $a_j$, for every $j \in [m]$, let $b_{j, 1}, b_{j, 2}, \ldots, b_{j, \ell_j}$ (note that $\ell_j = 0$ is possible) be exactly the elements already produced in phase $i$ of the enumeration procedure of $A_j$ in some previous iterations. In other words, for every $j \in [m]$, we have requested in applications of Step~\ref{StepOne} exactly the elements $b_{j, 1}, b_{j, 2}, \ldots, b_{j, \ell_j}, a_{j}$ from phase $i$ of the enumeration procedure of $A_j$ (note that $a_j = \bot$ is possible, which means that $a_j$ does not exist and therefore has not been requested). Moreover, this has happened in $\ell = \max\{\ell_j \mid j \in [m]\}$ previous (i.\,e., not counting the current one) iterations of the main loop of the enumeration procedure of algorithm $A$. In particular, this means that we have in phase $i$ so far only produced $\ell$ pairs with $i$ as left element.  \par
Let $K = \{b_{j, p} \mid j \in [m], p \in [\ell_j]\}$ and let $M = \{a_j \mid j \in [m]\}$. Since $\ell = \max\{\ell_j \mid j \in [m]\}$, there is at least one $j' \in [m]$ with $\ell_{j'} = \ell$ and therefore $|\{b_{j', 1}, b_{j', 2}, \ldots, b_{j', \ell_{j'}}\}| = \ell$. This is true since the elements $b_{j', 1}, b_{j', 2}, \ldots, b_{j', \ell_{j'}}$ are part of the $i^{\text{th}}$ phase of the enumeration of $A_{j'}$ and therefore must be distinct. Thus, $|K| \geq \ell$. Furthermore, we can choose $j'$ such that $a_{j'} \neq \bot$. This is the case since we have $a_j = \bot$ if and only if $\ell_j < \ell$ (i.\,e., phase $i$ of the enumeration of $A_j$ has already terminated) and, as observed above, there must be at least one $j \in [m]$ with $a_j \neq \bot$. \par
For every $b \in K$, either $(i, b)$ has been produced as output, or $b \in X$. Since so far we have only produced $\ell$ pairs as output, this directly implies that if $|K| > \ell$, then $X \neq \emptyset$, which means that we reach Step~\ref{StepFour} with $X \neq \emptyset$. If, on the other hand, $|K| = \ell$, then $K = \{b_{j', 1}, b_{j', 2}, \ldots, b_{j', \ell_{j'}}\}$, which also means that $a_{j'} \notin K$ since $b_{j', 1}, b_{j', 2}, \ldots, b_{j', \ell_{j'}}, a_{j'}$ is an enumeration of distinct elements provided by algorithm $A_{j'}$ (recall that $a_{j'} \neq \bot$, as observed above). Consequently, we also have $a_{j'} \notin Y$ and therefore $a_{j'}$ is added to $X$ in Step~\ref{StepThree}. Hence, we reach Step~\ref{StepFour} with $X \neq \emptyset$.
\end{proof}

Now, we give upper bounds for $\enumProbClass{\BTRPQ}$ and $\enumProbClass{\SRPQ}$ separately. 

\begin{thm}\label{BTRPQTheorem}
Semi-sorted $\enumProbClass{\BTRPQ}$ can be solved with delay $\bigO(\Delta(\DBD))$ (and without preprocessing). 
\end{thm}

\begin{proof}
Let $\DBD$ be a graph database, let $q = (x_1 \altop x_2 \altop \ldots \altop x_k)^*$ and $q' = (x_1 \altop x_2 \altop \ldots \altop x_k)^+$, where $x_1, x_2, \ldots, x_k \in \Sigma$. It can be easily seen that enumerating $q(\DBD)$ or $q'(\DBD)$ is the same as enumerating the reflexive-transitive closure $(E_{\DBD'})^*$ or the transitive closure $(E_{\DBD'})^+$, where $\DBD'$ is obtained from $\DBD$ by deleting all $x$-adjacency lists with $x \notin \{x_1, x_2, \ldots, x_k\}$. In~\cite{CaselEtAl2020_arxiv}, it is shown for general directed graphs $G = (V, E)$ how to enumerate $E^*$ and $E^+$ sorted by first coordinate (denoted by ``row-wise'') with delay $\bigO(\Delta(G))$. This approach translates to a semi-sorted enumeration and can be used on $\DBD$ in such a way that all $x$-adjacency lists with $x \notin \{x_1, x_2, \ldots, x_k\}$ are ignored (so without preprocessing). Thus, by using this procedure on $\DBD$, we can enumerate $q(\DBD)$ and $q'(\DBD)$ semi-sorted with delay $\bigO(\Delta(\DBD))$.
\end{proof}

\begin{thm}\label{SRPQTheorem}
Semi-sorted $\enumProbClass{\SRPQ}$ can be solved with preprocessing $\bigO(|\DBD|)$ and delay $\bigO(\degree(\DBD))$. 
\end{thm}

\begin{proof}
Let $\DBD$ be a graph database over $\Sigma$. We assume that $V_{\DBD} = [n]$ with $1 \preceq 2 \preceq \ldots \preceq n$ (see Lemma~\ref{makeWellFormedLemma}).\par
Let $q = (x_1 \altop x_2 \altop \ldots \altop x_k)$ with $x_1, x_2, \ldots, x_k \in \Sigma$. Then $q(\DBD) = \{(u, v) \mid (u, x_i, v) \in E_{\DBD}, 1 \leq i \leq k\}$. Obviously, we can in time $\bigO(|\DBD|)$ compute this whole set by moving trough the list for $V_{\DBD}$ and for every encountered node $u$ and every $i \in [k]$, we access the $x_i$-adjacency list for $u$ and for all $v$ it contains, we add $(u, v)$ to a list. This list contains exactly the elements from $q(\DBD)$ and is also semi-sorted. Thus, we can compute $q(\DBD)$ completely in the preprocessing and therefore enumerate it semi-sorted with constant delay.\par
Next, we assume that \[q = (x_1 \altop x_2 \altop \ldots \altop x_k) (y_1 \altop y_2 \altop \ldots \altop y_{k'})\] with $x_1, \ldots, x_k, y_1, \ldots, y_{k'} \in \Sigma$, for which the algorithm is more difficult. \medskip\\
\textbf{Preprocessing}: First, we construct $\DBtimesRE{\DBD}{q}$ in time $\bigO(|\DBD||q|) = \bigO(|\DBD|)$. Since $G_q$, the underlying $\Sigma$-graph of the $\NFA$ for $q$, can be assumed to have nodes $\{1, 2, 3\}$, where $1$ is the initial and $3$ the accepting state, and arcs $\{(1, x_i, 2) \mid 1 \leq i \leq k\} \cup \{(2, y_i, 3) \mid 1 \leq i \leq k'\}$, we can assume that the underlying non-labelled graph of $\DBtimesRE{\DBD}{q}$ has the following simple structure: $\DBtimesREV{\DBD}{q} = V_1 \cup V_2 \cup V_3$ with $V_i = \{(u, i) \mid u \in V_{\DBD}\}$ for every $i \in [3]$, and $\DBtimesREE{\DBD}{q} = E_1 \cup E_2$ with $E_1 = \{((u, 1), (v, 2)) \mid \exists i \in [k]: (u, x_i, v) \in E_{\DBD}\}$ and $E_2 = \{((u, 2), (v, 3)) \mid \exists i \in [k']: (u, y_i, v) \in E_{\DBD}\}$. The task of enumerating $q(\DBD)$ reduces now to the task of enumerating all reachable pairs $((u, 1), (v, 3))$ or, equivalently, to enumerate $(\DBtimesREE{\DBD}{q})^* \cap (V_1 \times V_3)$. We further assume that the sets $V_1$, $V_2$ and $V_3$ are stored in individual lists. This can be easily achieved during the construction of $\DBtimesRE{\DBD}{q}$: for every $u \in V_{\DBD}$ we add $(u, 1)$ to a list that stores $V_1$, then we construct lists for $V_2$ and $V_3$ analogously.\par
Finally, we modify $\DBtimesRE{\DBD}{q}$ as follows. We construct a Boolean array $S_{\text{in}}$ of size $|V_{\DBD}|$ (that can be addressed by the elements from $V_{\DBD}$) initialised with $0$ in all entries. Then, we move through the list for $V_1$ and for every $(u, 1)$ we encounter, we move through its adjacency list and for every $(v, 2)$ we encounter in this adjacency list, we set $S_{\text{in}}[v] = 1$. This procedure can be carried out in $\bigO(|\DBD|)$ and after its termination, we have that $S_{\text{in}}[v] = 0$ if and only if the in-degree of $(v, 2)$ is $0$. Next, we construct a Boolean array $S_{\text{out}}$ of size $|V_{\DBD}|$ (that can be addressed by the elements from $V_{\DBD}$) initialised with $0$ in all entries. Then, we move through the list for $V_2$ and for every $(u, 2)$ we encounter that has a non-empty adjacency list, we set $S_{\text{out}}[u] = 1$. This procedure can be carried out in $\bigO(|V_{\DBD}|)$ and after its termination, we have that $S_{\text{out}}[u] = 0$ if and only if the out-degree of $(u, 2)$ is $0$. We move again through the list for $V_2$ and we remove every encountered node $(v, 2)$ if $S_{\text{in}}[v] = 0$ or $S_{\text{out}}[v] = 0$; moreover, we move through the list for $V_1$ and for every $(u, 1)$ we encounter, we move through its adjacency list and we remove every $(v, 2)$ that we encounter if $S_{\text{in}}[v] = 0$ or $S_{\text{out}}[v] = 0$. These two steps can again be carried out in time $\bigO(|\DBD|)$. We have now removed all nodes from $V_2$ with in-degree or out-degree $0$, and all  arcs adjacent to such nodes. Next, we remove all isolated nodes from $V_{1}$ and $V_{3}$. This can be done by first computing all nodes from $V_1$ with out-degree $0$ and all nodes from $V_3$ with in-degree $0$, which can be done in the same way as we did for nodes from $V_2$ in time $\bigO(|\DBD|)$.\medskip\\
\textbf{Enumeration}: Let $Q$ be an empty queue. For every $i = 1, 2, \ldots, n$, we start a BFS from $(i, 1)$. Such a BFS is carried out until all nodes from the neighbourhood of $(i, 1)$ have been visited; thus, the BFS proceeds by visiting all the neighbourhoods $N((j, 2))$ for every $(j, 2) \in N((i, 1))$. From now on, we add $(i, k)$ to $Q$ for every node $(k, 3)$ that we visit for the first time. Whenever $\degree(\DBD)$ steps are made by the BFS after the last output while visiting a neighbourhood $N((j, 2))$, we produce after visiting this whole neighbourhood  
 the first element from $Q$ as output and remove this element from $Q$. In order to do this efficiently, we store the elements of $Q$ not just as a list but also maintain two lists of size $[n]$ over $\{1,\dots,n\}$ storing the elements currently in and the elements already produced and deleted from $Q$, respectively.  When the BFS terminates, i.\,e., all neighbourhoods of the nodes from $N((i, 1))$ have been visited, then we produce the remaining elements in $Q$ (and thus emptying it) and proceed with the BFS for the next node from $V_1$.\medskip\\
\textbf{Correctness}: The enumeration procedure goes through $n$ phases, where in each such phase we perform a BFS in some vertex $u \in V_1$ and produce only pairs of the form $(i, j)$, where $u = (i, 1)$. This means it is sufficient to show for each phase separately that we will produce all pairs of the form $(i, j)$ (where $u = (i, 1) \in V_1$ is the BFS-start-node of this phase) without duplicates and with the desired delay. Moreover, note that the enumeration is obviously semi-sorted. Further, if $Q$ is never empty when we request the next element, the delay is obviously in $O(\degree(\DBD))$ since visiting one neighbourhood  $N((j, 2))$ requires at most $\degree(\DBD)$ steps. \par
Let $u \in V_1$ and consider the phase of the enumeration procedure that performs the BFS in $u$. The BFS visits all nodes of the neighbourhood $N(u) = \{v_1, v_2, \ldots, v_h\}$ and marks them as visited, which can be done in time $\bigO(h) = \bigO(\degree(\DBD))$. Then it visits the complete neighbourhood $N(v_1)$, then the complete neighbourhood $N(v_2)$ and so on until all nodes $\bigcup^h_{i = 1} N(v_i)$ are visited. For each $i\in[h]$, the BFS performs $k_i=\sum_{j=1}^i |N(v_j)|$ steps to visit all neighbourhoods $N(v_1),\dots,N(v_i)$, while producing at least $q_i=\max\{|N(v_j)|\mid 1\leq j\leq i\}$ different elements for $Q$.
 To show that $Q$ is never empty, it suffices to show that $\sum_{j=1}^ik_j\leq q_i\degree(\DBD)$ for all $i$; note that we request an element from $Q$ exactly after $\sum_{j=1}^{r_i} k_j$ steps for some $r_1,\dots,r_x\in [h]$ where $1=r_1<r_2<\dots,r_x$  and $\sum_{j=r_i}^{r_{i+1}+1} k_j\geq \degree(\DBD)$ for each $i\in [x]$. 
Since the degree of $u$ is at most $\degree(\DBD)$, we know that $h\leq \degree(\DBD)$ which means 
$\sum_{j=1}^ik_j\leq i\max\{|N(v_j)|\mid 1\leq j\leq i\}\leq \degree(\DBD) q_i$.
\end{proof}

Finally, by using Lemma~\ref{alternationConstructionLemma}, we can plug together Theorems~\ref{BTRPQTheorem}~and~\ref{SRPQTheorem} in order to obtain the following upper bound for $\enumProbClass{\bigvee(\SRPQ \cup \BTRPQ)}$.

\begin{thm}\label{restrictedRPQEnumTheorem}
Semi-sorted $\enumProbClass{\bigvee(\SRPQ \cup \BTRPQ)}$ can be solved with preprocessing $\bigO(|q|^2|\DBD|)$ and delay $\bigO(|q|^2\Delta(\DBD))$.
\end{thm}

\begin{proof}
Due to Theorems~\ref{BTRPQTheorem}~and~\ref{SRPQTheorem}, both semi-sorted $\enumProbClass{\SRPQ}$ and semi-sorted $\enumProbClass{\BTRPQ}$ can be solved with preprocessing $\bigO(|q||\DBD|)$ and delay $\bigO(\Delta(\DBD))$. This obviously also means that semi-sorted $\enumProbClass{\SRPQ \cup \BTRPQ}$ can be solved with preprocessing $\bigO(|q||\DBD|)$ and delay $\bigO(\Delta(\DBD))$. Lemma~\ref{alternationConstructionLemma} now implies that semi-sorted $\enumProbClass{\bigvee(\SRPQ \cup \BTRPQ)}$ can be solved with preprocessing $\bigO(|q|^2|\DBD|)$ and delay $\bigO(|q|^2\Delta(\DBD))$.
\end{proof}

Since $q = \ta \ta \in \SRPQ$ is sufficient to express $\BMMProb$ as $\RPQ$-evaluation, Theorem~\ref{restrictedRPQEnumTheorem} implies that enumerating (the $1$-entries of) Boolean matrix products can be solved with linear preprocessing and $\bigO(n)$ delay, but, on the other hand, this also immediately implies a matching data complexity lower bound for the upper bound of Theorem~\ref{restrictedRPQEnumTheorem}.

\begin{thm}\label{shortRPQLowerBoundTheorem}
If $\enumProb(\SRPQ)$ can be solved with prep. $\bigO(|V_{\DBD}|^{3 - \epsilon} f(|q|))$ and delay $\bigO(|\degree(\DBD)|^{1 - \epsilon} f(|q|))$ for some function $f$ and $\epsilon > 0$, then the $\combBMMProb$-hypothesis fails.
\end{thm}

\begin{proof}
If $\enumProb(\SRPQ)$ can be solved with preprocessing $\bigO(|V_{\DBD}|^{3 - \epsilon} f(|q|))$ and with delay $\bigO(|\degree(\DBD)|^{1 - \epsilon} f(|q|))$ for some function $f$ and $\epsilon > 0$, then we can also compute $q(\DBD)$ in total time $\bigO((|V_{\DBD}|^{3 - \epsilon} f(|q|)) + (|q(\DBD)| |\degree(\DBD)|^{1 - \epsilon} f(|q|)))$. For the reduction of Lemma~\ref{BMMLowerBoundEvalLemma} this implies that we can solve $\combBMMProb$ in time $\bigO(n^{3 - \epsilon} + (n^2 n^{1 - \epsilon})) = \bigO(n^{3 - \epsilon})$.
\end{proof}

Compared to the full class of $\RPQ$s, the class $\bigvee(\SRPQ \cup \BTRPQ)$ is quite restricted. However, comprehensive experimental analyses of query logs suggest that quite restricted $\RPQ$s are still practically relevant: in the corpus of more than $50$ million $\RPQ$s analysed in~\cite[Table~$4$]{BonifatiEtAl2019}, roughly $50\%$ of the $\RPQ$s are from $\BTRPQ$ and another $25\%$ are of the form $q = x_1 x_2 \ldots x_k$. Many of these $\RPQ$s of the form $q = x_1 x_2 \ldots x_k$ also satisfy $k \leq 2$, which means that they are $\SRPQ$s~\cite{MartensPersonalCommunication}.

\section*{Conclusions}

In this work, we thoroughly investigated the fine-grained complexity of evaluating regular path queries. We focussed on what can be considered the most simple case of $\RPQ$s, i.\,e., the solution set contains only node pairs (and no witness paths) that are connected by \emph{arbitrary} paths (in contrast to simple path semantics, trail semantics etc.). As explained in the introduction, more powerful $\RPQ$s become computationally intractable, so this simple setting is suitable for a fine-grained complexity analysis. On the other hand, it can still be considered as a core functionality to be found in typical graph query languages. We considered the evaluation problems typically investigated in database theory (see Table~\ref{ProblemsTable}): the basic decision problems of Boolean evaluation and testing, and the function problems of computing and counting the solution set, and finally the enumeration variant. \par
For the non-enumeration variants, we were able to complement the upper bounds obtained by the product-graph approach by conditional lower bounds (see Table~\ref{ResultsTableNonEnum}). Hence, an algorithmic approach leading to strictly better upper complexity bounds than the product-graph approach seems to require unlikely improvements with respect to computing orthogonal vectors and Boolean matrices. \par
Our picture is much less complete for the enumeration variant of $\RPQ$-evaluation (see Table~\ref{ResultsTableEnum}). Although we produce many individual conditional lower bounds (which rule out many algorithms), it is still open whether enumeration with $\bigO(|\DBD|)$ preprocessing and $\bigO(|V_{\DBD}|)$ delay (or any delay truly sublinear in the size of the graph database) is possible (in data complexity); see also Question~\ref{openProblem}. Since finding an enumeration algorithm that achieves a sublinear delay is definitely a worthwhile research task (or, conversely, an enumeration algorithm with a delay bound as large as the whole database seems disappointing), we investigated several approaches to find enumeration algorithms with a sublinear delay. Our first such approach is to drop the restriction of linear preprocessing and it turns out that if we allow an additional factor of $\avgdegree(\DBD)$ in the preprocessing (or $\log(\avgdegree(\DBD))\avgdegree(\DBD)$ if we do not want to use lazy-initialisation), then a delay of $\bigO(|V_{\DBD}|)$ is indeed possible (see Theorem~\ref{crossProductEnumSuperlinearPreprocTheorem} and Corollary~\ref{crossProductEnumSuperlinearPreprocCorollary}).\footnote{Recall that $\avgdegree(\DBD)$ is the average degree of $\DBD$.} Our second approach shows that even linear preprocessing and constant delay is possible, if we are satisfied with enumerating just a representative subset of the solution set (see Theorem~\ref{apprEnumTheorem}). Here, \emph{representative} means that if $u$ can reach any node $v$ with a correctly labelled path, then at least one such witness pair exists in the set (and likewise if $u$ can be reached by any node $v$ with a correctly labelled path). While this variant is quite restricted in comparison to the full evaluation task, it might be a worthwhile query result to start with (e.\,g., before running a complete enumeration algorithm \emph{without} constant delay, we could first enumerate all these witnesses in constant delay and see whether this information is already enough). Finally, in our third approach, we identify a class of $\RPQ$s that can be enumerated with linear preprocessing and delay $\bigO(\Delta(\DBD))$ (see Theorem~\ref{restrictedRPQEnumTheorem}).\footnote{Recall that $\Delta(\DBD)$ is the degree of $\DBD$.} This result points out that the simplicity of the regular expression of the query might be exploited to achieve a better delay. This is particularly interesting given the fact that very simple $\RPQ$s are already sufficient for conditional lower bound reductions with respect to data complexity. Moreover, empirical work has shown that in practical scenarios where regular expressions are used as means of querying graph databases it is often the case that the regular expressions are rather simple. In case that Question~\ref{openProblem} can be answered in the negative, and if enumeration algorithms with sublinear delay are of high relevance, we should concentrate on algorithms that only work for a special and simple class of $\RPQ$s.\par
A possible future research task with respect to the topic of this paper is to answer Question~\ref{openProblem}. We conjecture that an algorithm that answers the question in the affirmative will be non-trivial and likely to yield more general algorithmic insights with respect to querying graphs with regular expressions. If, by a conditional lower bound, it can be shown that the answer to the question is negative, then the question arises for which $\RPQ$s a sublinear delay is possible and for which it is (conditionally) not. Our Theorem~\ref{restrictedRPQEnumTheorem} constitutes a partial result in this regard.

\section*{Acknowledgment}
\noindent We wish to thank the anonymous reviewers of the conference version~\cite{CaselSchmid2021} of this work for their valuable feedback. In particular, following the reviewer's comments and suggestions, we have included more background information and comprehensive explanations of certain aspects, which substantially improved the overall exposition of this paper. We also thank the reviewers of this journal version for their thorough reviewing.

\bibliographystyle{alphaurl}   % Was the commented out bibliography on purpose? Let us know if you want it changed
\bibliography{bibfile}

\iffalse
\newcommand{\etalchar}[1]{$^{#1}$}

\fi

\end{document}